\documentclass[10pt,journal,compsoc]{IEEEtran}

\ifCLASSOPTIONcompsoc
  \usepackage[nocompress]{cite}
\else
  \usepackage{cite}
\fi

\ifCLASSINFOpdf
\else
\fi
\usepackage{subcaption}
\usepackage{float}
\usepackage{amsfonts}
\usepackage{amsmath}
\usepackage{url}
\usepackage{caption}
\usepackage[linesnumbered,ruled]{algorithm2e}
\usepackage{algpseudocode}
\usepackage{array}
\usepackage[pdftex]{graphicx}
\usepackage{amsthm}
\usepackage{color}
\usepackage{enumitem}

\newtheorem{theorem}{Theorem}
\newtheorem{mydef}{Definition}
\newtheorem{corollary}{Corollary}
\newtheorem{lemma}{Lemma}

\newcommand{\Bstar}{$\bf{B_\Delta^\star}$}
\newcommand{\Dstar}{$\bf{C^\star}$}
\newcommand{\Estar}{$\bf{C_2^\star}$}
\newcommand{\Cstar}{$\bf{B_D^\star}$}

\def\th {\mathrm{th}}

\begin{document}
%
\title{Towards Optimal Grouping and Resource Allocation for Multicast Streaming in LTE}
%
%
%
%
%

\author{Sadaf~ul~Zuhra,~
        Prasanna~Chaporkar~
        and~Abhay~Karandikar} 
\IEEEtitleabstractindextext{%
\begin{abstract}
Multimedia traffic is predicted to account for $82$\% of the total data traffic by the year 2020. With the increasing popularity of video streaming applications like YouTube, Netflix, Amazon Prime Video,
popular video content is often required to be delivered to a large number of users simultaneously. Multicast transmission can be used for catering to such 
applications efficiently. 
The common content can be transmitted to the users on the same resources resulting in considerable resource conservation. 
This paper proposes various schemes for efficient grouping and resource allocation for multicast transmission in LTE.
The optimal grouping and resource allocation problems are shown to be NP-hard and so, we propose heuristic algorithms for both these problems. We also 
formulate a Simulated Annealing based algorithm to approximate the optimal 
resource allocation for our problem. The LP-relaxation based resource allocation proposed by us results in allocations very close to the estimated optimal.
\end{abstract}

}
\maketitle

\IEEEdisplaynontitleabstractindextext

%
\IEEEpeerreviewmaketitle

\IEEEraisesectionheading{\section{Introduction}\label{sec:introduction}}
Multicast transmission refers to one-to-many transmission from a single source to multiple receivers simultaneously. Today's cellular communication is primarily based on one-to-one communication or what we call unicast transmission. In unicast transmission, the evolved NodeB (eNB)
communicates with each User Equipment (UE) separately using different resources for each one of them. Using multicast transmission, multiple UEs can receive content on the same resources simultaneously.
It can be effectively used for applications like video streaming from popular platforms such as YouTube, Netflix and Amazon Prime, streaming of television (TV) programs, software updates, news
updates and weather forecasts in which the same content is required to be transmitted to a large number of UEs simultaneously. Assigning orthogonal resources to every UE in this scenario is a very inefficient manner of resource allocation. 
Using multicast transmission for such applications can save considerable resources in a cell. Multicast attempts to transmit the common content using as few resources as possible so that the remaining resources can be used to simultaneously 
support other services in the cell. \par

For successfully implementing multicast transmission in Long Term Evolution (LTE), there are two main challenges that need to be addressed. The first is the problem of dividing UEs into different groups. The set of UEs grouped together forms what we will
henceforth refer to as a multicast group. The UEs in a multicast group are treated as a single entity by the eNB and can be served using the same resources. The way UEs are grouped together depends
on the criteria used for grouping. One evident requirement is for all the UEs in a group to require the same content. But, as we shall see later, grouping all the UEs into a single group only based on the criterion that
they need the same content,
may lead to a degraded system performance due to varied channel gains experienced by different UEs. Therefore, the channel gains of the UEs also need to be considered while dividing them into multicast groups.
The second problem to be addressed is that of the resource allocation to the multicast groups. Since we aim at minimizing the resources used for multicast operations, a resource allocation scheme needs to be designed accordingly. \par

Provisions for multicast transmission in LTE have been introduced in Release $9$~\cite{tenth} of Third Generation Partnership Project (3GPP) standards by inclusion of Multimedia Broadcast Multicast Services (MBMS).
In LTE, resources are divided on a time and frequency scale. A radio frame spans over a period of $10$ ms and consists of $10$ sub-frames of $1$ ms each. A sub-frame is made up of smaller units termed as the Physical Resource Blocks (PRBs).
A PRB is the smallest unit of allocation in LTE. 
MBMS allows for point-to-multipoint transmission so that the eNB can transmit to multiple UEs using the same PRBs~\cite{mbmsrel14}.
UEs can subscribe to a MBMS service and are notified when a MBMS session is going to start. They can then receive the relevant content from the eNB. All UEs subscribed to a particular MBMS service are served using common resources. Serving all the 
subscribed UEs using the same PRBs is however, not very efficient. Treating all the multicast UEs as a single group makes the performance of MBMS dependent on the channel gain of the weakest UE in the group. This may significantly bring down the quality of service 
for the other UEs in the group that may be experiencing a much better channel with the eNB. Thus, making groups based solely on the content requirement can lead to a degraded system performance. Let us consider an example to illustrate this. \par

Consider a sub-frame in which $10$ PRBs are available for allocation. Two UEs in the cell, $U_1$ and $U_2$ have subscribed to the same MBMS service. Let the required minimum rate for this service be $10^3$ bits/sub-frame.
When a PRB is allotted to a set of MBMS UEs, the rate at which reliable transmission can take place corresponds to the UE with the least channel gain in the group. Transmitting at a rate greater than this can lead to unsuccessful decoding by the UEs
with the least channel gain. Consider a state where $U_1$ has a good channel in all odd numbered PRBs so that as many as $10^3$ bits can be transmitted in each of them. In the rest
of the PRBs, assume that $U_1$ can only get a maximum of $100$ bits each. Similarly assume that $U_2$ can receive a maximum of $10^3$ bits in each of the even numbered PRBs and $100$ bits in the odd numbered PRBs. 
Now, if these UEs are grouped together for MBMS multicast transmission, data can be transmitted at a rate corresponding to the UE with the least channel gain (and hence the least rate)
in each PRB. So, in this case, only $100$ bits can be transmitted in each PRB and to satisfy the minimum rate, all $10$ PRBs will be used. On the other hand, if $U_1$ and $U_2$ are grouped separately,
they can be alloted PRB $1$ and PRB $2$ respectively and $10^3$ bits can be transmitted in each of these PRBs. Thus, the required rate for both will be satisfied in just $2$ PRBs, $8$ less than the previous scenario. 
This example shows that appropriate grouping of UEs who have subscribed to a given MBMS service is essential for obtaining any benefit whatsoever from multicast operations.  \par

In this paper~\footnote{This work is an extension of the work done by us in~\cite{self}}, we consider multicast transmission with UEs having different channel gains in each PRB.
When a large number of UEs in a cell require the same content, a MBMS session can be initiated to cater to them using multicast transmission.
We shall now briefly discuss the main challenges involved in the grouping and the resource allocation problems for multicast transmission.
The main challenge in grouping UEs based on their channel gains is that, due to fast fading, the channel gains experienced by the UEs keep on changing. 
However, grouping UEs based on their instantaneous Signal-to-Noise Ratio (SNR) in every sub-frame is also not feasible as it can lead to increased control overhead due to frequent changes in grouping. 
Since each multicast group is treated as a separate entity by the eNB, each group is assigned a unique MBMS Radio Network Temporary Identifier (M-RNTI). M-RNTI of a group is used for scrambling it's Downlink Control Indicator (DCI) 
which carries the resource allocation information in LTE~\cite{rnti}. If grouping is changed every sub-frame, a new M-RNTI will have to be assigned and conveyed to the UEs every sub-frame.
Therefore, the grouping policies need to balance these two factors so that the grouping is efficient and need not be changed every sub-frame.
In addition to this, grouping policies need to answer key questions like the number of groups that should be formed or the maximum number of UEs that should be placed in a group. Creating a lesser number of groups means more number of UEs in a single group and may result in a lesser number of PRBs being used by the eNB. However, as the number of UEs in a group increases, the probability that at least one UE is in deep fade also increases. This leads to a degraded system performance. On the other hand, a larger number of groups means more number of entities to be served by the eNB which requires a larger number of resources. Thus, there is a trade-off between the number of multicast groups to be formed and the number of UEs in each group which needs to be balanced by a grouping policy. \par
Once the groups are formed, in each sub-frame, resources have to be allocated to the groups. The aim of the resource allocation problem is to minimize the PRBs allocated to multicast UEs while guaranteeing a certain Quality of Service.
The optimal resource allocation problem can be formulated as a Binary Linear Program (BLP) with the objective of minimizing the number of 
PRBs used to cater to the multicast services constrained to satisfying the rate requirements of all the groups. BLPs are inherently hard to solve and require significant computational power even for small input sizes.
Thus, the optimal grouping and optimal resource allocation problems are non-trivial and we need to design efficient algorithms for their successful implementation.
In this paper, we address both these problems and formulate algorithms to overcome the discussed challenges.
Next, we discuss some of the literature relevant to the problems under consideration.

\subsection{Related Literature}

The literature related to multicast group formation and resource allocation can be broadly classified into four categories namely opportunistic scheduling schemes, joint optimization schemes for unicast and multicast, sub-group formation schemes and schemes for multicasting of Scalable Video Coded (SVC) content. We now describe the work done in each of these areas.
\subsubsection{Opportunistic multicast scheduling:}
Opportunistic scheduling schemes, as the name suggests, are throughput maximizing scheme that schedule UEs with the best channel conditions in a sub-frame.
In~\cite{fifth}, the authors present an optimized version of Opportunistic Multicast Scheduling (OMS) that balances between multicast gain and multi-user diversity. A fraction of UEs with the best channel gains are scheduled in each time slot.
\cite{mbsfn_om} studies the use of opportunistic multicasting for Single Frequency Networks (SFNs). The authors focus on maximizing the spectral efficiency of the SFNs by
opportunistic scheduling of UEs with higher Channel Quality Indicator (CQI) values. \par 
 In~\cite{fdps}, the authors propose a Frequency Domain Packet Scheduler (FDPS) for MBMS that maximizes the minimum rate achievable by
UEs in a PRB. It uses a somewhat pessimistic approach in that it only minimizes the performance loss caused by the worst PRB assignment. 
Moreover, the performance of the proposed policy has only been compared to a blind FDPS policy. The blind FDPS uses a
blind static allocation that doesn't change over time which is not a good benchmark to compare with. \par
In~\cite{genetic}, the authors propose the use of a genetic algorithm for resource allocation in OFDMA multicast followed by power allocation based on the technique proposed in~\cite{genetic_power}. The resource allocation problem in~\cite{genetic} aims to maximize the total throughput subject
to power and fairness constraints. The authors, however, do not subgroup the UEs based on their channel states. All UEs receiving the same content are put into a single multicast group. 
Recently, there has also been some work on multicast transmission for 5G satellite systems. In~\cite{5gsat_1} and~\cite{5gsat_2}, the authors propose solutions for radio resource management and subgrouping for multicast over 5G satellite systems. The optimization problems formulated seek to maximize the Aggregate Data Rate (ADR) of the system. \cite{5gsat_1} and \cite{5gsat_2} also assume a single CQI value corresponding to a multicast UE which means that all PRBs are equivalent for a UE.
Maximizing the ADR is also the objective function of \cite{bargain} that makes use of game theoretic bargaining solutions for grouping and resource allocation of 
multicast UEs. \par
Most of the literature considers only wideband CQI (i.e. a single CQI value for the entire available bandwidth) for grouping and resource allocation in multicast transmission. \cite{subband} is one work that explores the use of subband CQI values in multicast resource allocation. The objective function here is still the maximization of the ADR as in \cite{5gsat_basic}, \cite{5gsat_1} and \cite{5gsat_2}. However, with the consideration of different subband CQI values, a closed form solution for subgroup formation as given in \cite{5gsat_basic} no longer remains feasible. 
While all the papers mentioned in this section seek to maximize the ADR in some way, in this paper, we focus on providing a certain minimum rate to every multicast UE based on the service
it is subscribed to.

\subsubsection{Joint optimization for unicast and multicast:}
This section summarizes the literature that deals with the problems of joint resource allocation to unicast and multicast UEs.
In \cite{joint_2}, \cite{joint_3} and \cite{joint_power} joint delivery of unicast and multicast/broadcast transmission in LTE and OFDMA systems has been addressed. Policies proposed in \cite{joint_3} and
\cite{joint_power} guarantee a certain rate to all the UEs and make use of unicast transmission for serving UEs with the worst CQI values. In \cite{joint_2},
the performance of streaming over MBSFNs and file delivery over eMBMS has been evaluated through simulations. Performance indicators like outage probability, coverage and maximum supportable
MCS have been used to asses the feasibility of various MBMS configurations from the perspective of the service providers. \par
In \cite{mung_chiang}, the authors deal with resource allocation in eMBMS. The authors assume that the video content is simultaneously available through unicast as well as 
eMBMS and the primary problem seeks to jointly optimize over the grouping of UEs and allocation of resources to unicast and eMBMS. The resource allocation scheme 
proposed in the paper allocates resources to groups proportional to the number of UEs in a group.
None of these papers consider the varying channel conditions of UEs over different PRBs while allocating resources. In this paper, however, we account for the fact that the
CQIs of UEs may be different in every PRB of a sub-frame.

\subsubsection{Sub-group formation for multicast:}
In this section, we summarize the literature that primarily deals with dividing multicast UEs that require the same content into multiple sub-groups.
In~\cite{EGB}, the authors deal with the grouping problem for MBMS in High Speed Packet Access (HSPA) networks. They propose a grouping policy that minimizes a `Global Dissatisfaction Index' (GDI). GDI accounts for the difference in 
the maximum data rates achievable by UEs and the rates actually assigned to them. The authors show, using simulations, that their proposed policy performs better in terms of UE satisfaction compared to MBMS transmission without grouping.
In~\cite{first}, the same authors investigate the effect of pedestrian mobility on the performance of the grouping policy proposed in~\cite{EGB}. \par
In~\cite{rrm_ltea}, the authors propose subgrouping and resource allocation for multicast in LTE-A systems. Extensions of LTE multicast subgroup formation are presented for use
in LTE-A systems. They propose a radio resource management scheme that achieves a trade-off between efficiency and fairness. For resource 
allocation, they make use of the bargaining solutions proposed in~\cite{bargain}. It is shown that, due to carrier aggregation
in LTE-A systems, the overall system throughput is significantly increased but the relative performance of the studied algorithms remains the same. Extension of bargaining solutions proposed in \cite{bargain} to multi-carrier systems like LTE-A have been studied in \cite{bargain_ltea}. In \cite{freq_select}, the authors have extended the work from \cite{bargain} to exploit frequency selectivity for improving the spectral efficiency of multicast in LTE.
\cite{hetnet} and \cite{second} deal with the use of multicast in heterogeneous networks and grouping of UEs for MBMS respectively.
Low complexity variations of the Subgroup Merging Scheme (SMS)~\cite{sms} that provide better ADR have been proposed in~\cite{low_subgroup} for improving scalability. \par
The papers mentioned in this section use the entire set of PRBs for catering to the multicast UEs. We, however, aim to satisfy the multicast UEs in the minimum possible
number of PRBs as in a practical scenario, an eNB has to support multiple other services along with multicast sessions.

 \subsubsection{Multicasting of SVC content:}
 In this section, we present a summary of the literature in which various problems related to multicasting of SVC video streams have been studied.
\cite{mood} and \cite{mood_2} deal with resource allocation for MBMS Operation On-Demand (MooD) for video streams. The authors consider Quality of Experience (QoE)
instead of Quality of Service (QoS) as the utility function that is maximized by the resource allocation schemes.
\cite{power} examines power efficient video streaming via MBMS. The authors study the multicast streaming of high quality SVC encoded videos. The UEs are grouped together based on
the content, the quality of content requested and their physical proximity. The algorithms proposed try to minimize the power consumption by sending traffic in discontinuous bursts,
allowing UEs to sleep in between bursts.
In~\cite{topsis}, the subgrouping and resource allocation decisions are based on three criteria, maximizing the throughput, maintaining fairness and minimizing the 
dissatisfaction of groups. The authors make use of Technique for Order of Preference by Similarity to Ideal Solution (TOPSIS)~\cite{topsis_93}, a method for multi-criteria decision 
making in subgrouping and resource allocation algorithms. TOPSIS has also been used in~\cite{evaluation} for comparing the performance of various multicast resource allocation schemes based on their ADR, fairness and spectral efficiency.
Even though SVC provides an interesting new method of video encoding with various benefits, H.264/AVC continues to be the choice of encoding videos over the Internet. Most of the popular streaming platforms like 
Netflix~\cite{netflix} and YouTube use H.264/AVC to encode their videos. \par
To the best of our knowledge, the problem of satisfying the rate requirement of each MBMS UE in each time slot while minimizing the PRB utilization at the eNB has not been addressed in the existing literature.
This is a very important problem because the practical success of the multicast services strongly depends on how well they can co-exist with the other extremely large number of services 
supported by LTE and the next generation 5G networks~\cite{multicast_5g}. While MBMS services are suitable for real time streaming applications, their resource utilization 
has to be such that sufficient resources are available for the non real time applications being simultaneously provided in the cells. Since the resources used 
for providing an MBMS service are essentially disseminating the same content to the MBMS UEs, over provisioning of resources for multicast must be avoided. The multicast 
UEs could simultaneously be using unicast services along with other UEs in the cell which may not be involved in the ongoing MBMS sessions. Minimizing the resources used 
by the MBMS services will ensure that the impact of multicast services on the rest of the operations in the LTE cell is minimized. \par 

In most of the literature, the rate achievable by a UE is assumed to be the same over all PRBs. This means that the achievable rates of UEs are same in every PRB. This assumption greatly simplifies the resource allocation problem as the identities of the PRBs are no longer important. In practice, however, the channel response can be different for different frequency channels resulting in varying channel gains over different PRBs. In this work, we take these variations into account. Thus, the channel
states and hence the CQI values for a group or user vary over different PRBs in a sub-frame.

A large portion of the literature including~\cite{genetic} and~\cite{joint_power} claim that the grouping and resource allocation problems are `hard to solve' or 
`infeasible'. However, none of these papers present any mathematical 
proof of hardness of the grouping or resource allocation problems for multicast transmission. In this paper, for the first time, we present the proof of NP-hardness of both
the optimal grouping as well as the optimal resource allocation problem. Next, we summarize the main contributions of this paper:\\
$\bullet$ We prove that the optimal resource allocation problem that minimizes the number of PRBs utilized while providing a minimum rate to all the multicast groups is a NP-hard problem
and so no polynomial time algorithm can be formulated for determining it's optimal solution unless P = NP. \\ 
$\bullet$ The optimal grouping problem is also shown to be a NP-hard problem. \\ 
$\bullet$ We devise a randomized scheme for estimating the optimal resource allocation. The randomized scheme works iteratively to end up at the optimal
solution with high probability. The output of this scheme acts as a benchmark for the heuristic schemes for resource allocation. \\ 
$\bullet$ We propose two heuristic schemes for resource allocation to multicast groups, a greedy scheme and a LP-relaxation based scheme. \\ 
$\bullet$ We also propose two heuristic schemes for multicast group formation, a fixed size grouping scheme and a CQI based grouping scheme. \par 


The rest of this paper is organized as follows. In Section~\ref{sec:sys_model}, we discuss the problem formulation and the system model.
The Simulated Annealing (SA) based randomized scheme and related results are presented in Section~\ref{sec:RAS}. In Section~\ref{sec:RASh}, we present the proposed heuristic schemes for 
resource allocation. The proposed heuristic schemes for grouping are discussed in Section~\ref{sec:grouping}. 
We present the simulation results in Section~\ref{sec:simulations} and conclude in Section~\ref{sec:conclusion}.
In the interest of preserving the flow of the paper, proofs of NP-hardness of the optimal grouping and optimal resource allocation problems 
are presented in Sections~\ref{NPC1} and~\ref{NPC2} respectively. The notations used in the paper have been summarized in Table~\ref{notations} for easy reference of the reader.

\section{Problem Formulation} \label{sec:sys_model}
We consider a single LTE cell with $M$ multicast UEs. All the UEs have subscribed to the same MBMS service and require to be served at a minimum rate of $R$ bits/sec. The required rate can 
be provided to each UE by allotting one or more PRBs in each LTE sub-frame. We denote the number of PRBs in a sub-frame by $N$. Let $[n] = \{1,\ldots,n\}$ and let $|A|$ denote the cardinality of a set $A$. Thus, $[M]$ and $[N]$ denote the set of all
multicast UEs and the set of PRBs in a sub-frame, respectively. We assume that the channels between eNB and the multicast UEs are location and
time varying. Thus, each UE has different channel gains in different PRBs and also across different sub-frames. We assume block fading channel model, and hence
the channel gain of a UE is assumed to remain the same across a sub-frame. Though we are not considering mobility explicitly, our approach can be extended to cases where
UE positions evolve at a much slower time scale than the sub-frame duration. Let $h_{iu}[t]$ denote the channel gain for UE $u$ on $i^\th$ PRB in sub-frame $t$.
$h_{iu}[t] = \overline{h}_{iu} + H_{iu}[t]$, is made up of $2$ components. $\overline{h}_{iu}$ denotes the average channel gain which accounts for path loss and shadowing and is invariant across sub-frames.
$H_{iu}[t]$ is the fast-fading component that varies across sub-frames. $H_{iu}[t]$'s are independent and identically distributed (i.i.d) exponential random variables.

We assume that the eNB has full Channel State Information (CSI) of all the UEs. Corresponding to the channel gain, the eNB assigns the maximum supportable rate, $r_{iu}[t]$
bits/sec for UE $u$ on $i^\th$ PRB in sub-frame $t$. Note that $r_{iu}[t]$ is determined by the Modulation and Coding Scheme (MCS) used, and thus can take finitely many values (15 as per current standards for LTE~\cite{seventh}). 
Next, we discuss grouping.\par

Since all multicast UEs want the same content in each sub-frame, the UEs can be grouped together and
served on common PRBs. A grouping strategy, $\Delta$ is defined as follows:

\begin{mydef}
A grouping strategy $\Delta$, defines a partition
$\{G_1^\Delta,\ldots,G_L^\Delta\}$ of $[M]$, where $G_i^\Delta \subseteq [M]$ is referred to as the $i^\th$ group.
\end{mydef}

Note that $L \le M$ and when $L=M$, we have the unicast case. Henceforth, unicast is not dealt with separately. Throughout this paper, we assume that the groups once defined at the beginning of a MBMS session cannot be changed in that session.
This is done to avoid excessive control overhead that may result due to rapid changes in grouping.
One can relax the assumption and allow for grouping to be potentially changed every $K$ sub-frames, where $K$ is large. This will allow the scheme to adapt in case of mobile networks.
The minimum supportable rate for a group $G_j$ on $i^\th$ PRB in sub-frame $t$ ($r_{ij}^\Delta[t]$) is equal to the minimum of the rates achievable by it's constituent members, i.e., $r_{ij}^\Delta[t] = \min_{u\in G_i^\Delta}\{r_{iu}[t]\}$.
This is to ensure that the content received by the group can be successfully decoded by all the members. If we transmit at rates more than this, the weakest UE in the group may not be able to decode the received content successfully.
Once the $r_{ij}^\Delta[t]$'s are obtained, we need to decide how resources will be allotted to each group so that the total number of PRBs used is
minimized subject to giving each group at least the minimum required rate $R$. This is a resource allocation problem. The formal definition of a resource allocation policy is stated below.
\begin{mydef}
For a given grouping $\Delta$ a resource allocation policy, $\Gamma$ defines an assignment of PRBs to the $L$ multicast groups, \{$\overline{V}_{1\Gamma}^\Delta, \ldots , \overline{V}_{L\Gamma}^{\Delta}$\}, where, $\overline{V}_{i\Gamma}^\Delta$ is the
set of PRBs assigned to group $i$ by resource allocation policy $\Gamma$ under grouping $\Delta$. The allocation $\Gamma$ should be such that $\bigcap_{i=1}^L\overline{V}_{i\Gamma}^\Delta = \phi$ and $\bigcup_{i=1}^L\overline{V}_{i\Gamma}^\Delta \subseteq [N]$.
\end{mydef}
The resource allocation policy $\Gamma$ is said to be feasible if $\sum_{j \in \overline{V}_{i\Gamma}}^{\Delta}r_{ij}^{\Delta}[t] \geq R$ for every $i \in [L]$.
The other parameter used by us to characterize a resource allocation policy is the number of PRBs left unused after resources have been allocated in a sub-frame $t$, $S_{\Gamma}^\Delta[t]$.
So, $S_\Gamma^\Delta[t] = N - |\bigcup_{i=1}^L\overline{V}_{i\Gamma}^\Delta|$.
We shall now formally state our resource allocation and grouping problems.

\subsection{Problem 1: Optimal Resource Allocation \Bstar}
Consider a fixed grouping policy $\Delta$, and define indicators in sub-frame $t$ as follows:
\begin{eqnarray*}
x_{ij}[t] = \begin{cases}
	{1}, & \text{if PRB $j$ is assigned to group $i$}\\
	{0}, &  \text{otherwise}.
\end{cases}
\end{eqnarray*}
The optimal resource allocation can be obtained as a solution to the following BLP for every $t$:
\begin{eqnarray}
(\bf{B_\Delta^\star}):\ \ \label{eq:objective}\min\sum_{j\in[N]}\sum_{i\in [L]}x_{ij}[t],  \nonumber\\
\label{eq:rate_constr}\text{subject to:} \ \ \sum_{j\in[N]} x_{ij}[t] r_{ij}^\Delta[t] &\geq& R, \hspace{5mm} 
\forall \ i \in [L],\\
\label{eq:matching_constr}\sum_{i\in[L]}x_{ij}[t] &\leq& 1, \hspace{5mm} \forall \ j\in[N].
\end{eqnarray}
The objective function of \Bstar seeks to minimize the number of utilized PRBs in sub-frame $t$.
Constraint (\ref{eq:rate_constr}) guarantees that the total rate given to each group is greater than or equal to the desired minimum rate $R$ and (\ref{eq:matching_constr}) ensures that each PRB is given to at most one group.

Note that \Bstar \ gives the optimal resource allocation for any grouping $\Delta$. Next, we establish the hardness of \Bstar.

\begin{lemma} \label{NPC1_lemma}
 Optimization \Bstar \ is NP-hard.
\end{lemma}
\begin{proof}
 Refer to Section~\ref{NPC1} for the detailed proof.
\end{proof}

\subsection{Problem 2: Optimal Grouping \Dstar}
Recall that $S_\Gamma^\Delta[t]$ denotes the number of PRBs left unutilized under grouping policy $\Delta$ in sub-frame $t$ using resource allocation scheme $\Gamma$. Note that these PRBs can be used for other UEs in the system.
Define,
\begin{align} \label{avg_metric}
\overline{S}_\Gamma^\Delta = \lim\inf_{\mkern-25mu T\to\infty} \frac{1}{T} \sum_{t=1}^T S^\Delta_\Gamma[t].
\end{align}
Thus, $\overline{S}^\Delta_\Gamma$ is the average number of unutilized PRBs per sub-frame under grouping policy $\Delta$ and resource allocation policy $\Gamma$. 
The optimal grouping problem can be defined for any given resource allocation policy $\Gamma$. 
The definition of the optimal grouping problem is stated below: \par
(\Dstar) : Determine the optimal grouping policy $\Delta^\star$ such that $\overline{S}_\Gamma^{\Delta^\star} \geq \overline{S}_\Gamma^\Delta$ for every $\Delta$. \par

We note that determining $\overline{S}_\Gamma^\Delta$ for a general grouping $\Delta$ and resource allocation $\Gamma$ itself is a very hard, if not an impossible problem. The value of $\overline{S}_\Gamma^\Delta$ depends on the combined channel states of all 
the UEs in various sub-frames. Hence, to examine the hardness of \Dstar, we consider a scenario where, for given $\Gamma$ and proposed $\Delta$, there exists a Genie that provides the value of $\overline{S}_\Gamma^\Delta$. We assume that the mapping from
$\Delta$ to $\overline{S}_\Gamma^\Delta$ can be an arbitrary positive valued function. We show in the following result that there exist mapping functions for which the problem of determining $\Delta^\star$ for given $\Gamma$ is NP-hard.

\begin{lemma} \label{NPC2_lemma}
 For a fixed $\Gamma$, the problem of determining $\Delta^\star$ for an arbitrary mapping function from the set of all groupings to the average number of unused PRBs is NP-hard.
\end{lemma}
\begin{proof}
 Refer to Section~\ref{NPC2} for the detailed proof.
\end{proof}


 Since we have proved that both optimal grouping and optimal resource allocation problems are NP-hard, no polynomial time algorithms exist for determining their optimal 
 solutions unless P = NP. We can, however, use some intelligent heuristic schemes to obtain near optimal solutions.
 In the following section, we formulate an iterative randomized scheme for estimating the optimal resource allocation. 
\begin{center}
\begin{table}
\captionof{table}{Table of notations} \label{notations} 
\begin{center}
\begin{tabular}{ | m{2 cm} | m{6 cm}|} 
 \hline
 \textbf{Notation} & \textbf{Meaning}\\
 \hline
 \hline
 $M$ & Number of multicast UEs \\
 \hline
 $L$ & Number of multicast groups \\
 \hline
 $N$ & Number of PRBs in a sub-frame \\
 \hline
 $[n]$ & $\{1,2, \ldots ,n\}$ \\
 \hline
 $|A|$ & Cardinality of set $A$ \\
 \hline
 ${\cal N}$ & Set of available PRBs in a sub-frame \\
 \hline
 ${\cal L}$ & Set of multicast groups \\
 \hline
 $h_{iu}[t]$ & Channel gain of UE $u$ on $i^{th}$ PRB in sub-frame $t$ \\
 \hline
 $r_{iu}[t]$ & Maximum rate supportable by UE $u$ on $i^{th}$ PRB in sub-frame $t$ \\
 \hline
 $\Delta$ & Grouping strategy \\
 \hline
 $\Delta^\star$ & The optimal grouping policy \\
 \hline
 $G_i^\Delta$ & $i^{th}$ group under policy $\Delta$ \\
 \hline
 $\Gamma$ & Resource allocation policy \\
 \hline
 $\overline{V}_{i\Gamma}^\Delta$ & Set of PRBs assigned to $G_i^\Delta$ under policy $\Gamma$\\
 \hline
 $R$ & Rate requirement of the multicast UEs \\
 \hline
 $S_\Gamma^\Delta[t]$ & Number of PRBs left unutilized under $\Delta$ in sub-frame $t$ using $\Gamma$ \\
 \hline
 $x_{ij}[t]$ & Indicator random variable that equals $1$ when PRB $j$ is assigned to group $i$ in sub-frame $t$ \\
\hline 
 \end{tabular} 
 \end{center}
 \end{table}
\end{center}

\section{Randomized Algorithm for Optimal Resource Allocation} \label{sec:RAS}
As stated in the previous section, no polynomial time algorithm exists for determining the optimal resource allocation. We can, however, estimate the optimal solution using randomized algorithms that iteratively explore the possible solutions
to finally converge to the optimum. The proposed randomized scheme serves dual purpose, $1$) it provides near optimal solution in much lesser computational power than that required to solve the BLP \Bstar and, $2$) it's output can be used as a benchmark for evaluation of the 
heuristic schemes which we propose in Section~\ref{sec:RASh}. Next, we describe a randomized algorithm for resource allocation. \par

The allocation of resources in LTE is done in every sub-frame. So, for brevity, we fix a sub-frame $t$ and omit it from notations in this section. Grouping strategy $\Delta$ impacts resource allocation via $r_{ij}^\Delta$, which is the rate 
achievable by group $i$ in PRB $j$. Here, we deal with resource allocation for any given $\Delta$. So, we omit $\Delta$ from the notations as well for better readability. 

The Randomized Scheme (RS) used here is based on SA, a well known Markov Chain Monte Carlo (MCMC) technique~\cite{ross}.
SA is a randomized algorithm used for obtaining the global optimum of a function. In SA, we use a Markov chain on the states of the problem under consideration and transition among the states to ultimately end up at the global optimum.
In our case, states correspond to all possible resource allocations to the groups. Therefore, every state, $s_d$ of the Discrete Time Markov Chain (DTMC) is a possible distribution of PRBs, \{$\overline{V}_{0d}, \overline{V}_{1d}. \ldots ,\overline{V}_{Ld}$\} 
where $\overline{V}_{id}$ is the set of PRBs assigned to group $G_i$, $i \in \{0,1, \ldots,L\}$ in state $s_d$. $G_0$ here is a dummy group that is assigned all the unused PRBs. Thus, the state space, $\chi$ corresponds to all possible PRB allocations to groups.
Let $\ell_{di}$ denote the total rate achieved by the $i^\th$ group in allocation $s_d$. Thus, $\ell_{di} = \sum_{j \in \overline{V}_{id}} r_{ij}$. Moreover, let $q_d$ denote $|\{i : \ell_{di} \geq R\}|$, i.e. $q_d$ is the number of satisfied groups in
allocation $s_d$. Each state has an associated reward that defines how good or bad the state is. For our DTMC, we define the reward function, $E$ from $\chi$ to the set of real numbers as follows:

\begin{equation*}
	E(s_d) = |\overline{V}_{0d}| - \sum_{i=1}^L \left[ R - \ell_{di} \right]^+ + q_d,
    \end{equation*}
where $[y]^+ = \max \{y,0\}$ and $|\overline{V}_{0d}|$ represents the number of unused PRBs in state $s_d$. The reward function is a monotonically 
increasing function of the number of satisfied groups and the number of unused PRBs. It also decreases proportionally with the difference between the required and achieved rates of the groups. Thus, intuitively, maximizing $E$ will maximize the number of 
unused PRBs while satisfying all the groups. We prove this formally in the next result.

\begin{lemma} \label{reward}
 Let \Bstar \ has a feasible solution and $s_{d^\star} \in \arg\max_{s_d} E(s_d)$. Define $x_{ij}^\star = 1$ if $j \in \overline{V}_{id^\star}$ and $0$ otherwise. Then, $\{x_{ij}^\star\}_{i,j}$ is the optimal solution of the BLP \Bstar.
\end{lemma}
\begin{proof}
 The detailed proof is given in Section~\ref{Eproof}.
\end{proof}
Thus, determining a state that maximizes the reward function is equivalent to determining the optimal solution of \Bstar.
 Note that the proposed approach uses a DTMC on $\chi$ where $|\chi| = (L+1)^N$. Recall that $L$ denotes the number of groups and $N$ denotes the number of PRBs available in a sub-frame. Hence, the Transition Probability Matrix (TPM) corresponding to the DTMC
 will have dimensions exponential in $N$. So, for guaranteeing computational feasibility of the proposed approach, one must ensure that the TPM need not be stored, rather, given the current state, transition to the next state can be determined in time polynomial in
 system parameters. Next, we elaborate how such a DTMC can be constructed.

 \subsection{DTMC Construction}
 Let $E^\star$ denote the maximum of the reward function, $E(.)$, i.e. $E^\star=\max_{s_d}E(s_d)$. We construct a DTMC $\{X_n^T\}_{n \geq 1}$ on $\chi$ such that $\mathbb{P}(E(X_n) = E^\star)$ tends to $1$ as $n$ tends to $\infty$.
 This implies that, if we simulate the DTMC for a large enough time, say $\tau$, then the probability that the state of the DTMC at time $\tau$ yields the optimal resource allocation is very close to one. Thus, there is a high chance of determining the optimal resource allocation
 using this randomized algorithm. \par
 
 Constructing a DTMC requires defining it's states, neighboring states, and the TPM. We have defined the states of the DTMC above. Now, we define the remaining terms.
 
 \subsubsection{Neighboring States} \label{neighboring_states}
 Consider any state $s_d \in \chi$. A state $s_{d'}$ is a neighbor of $s_d$ if it can be obtained from $s_d$ using one of the following actions:
   
	  $\bullet$ \textit{Swap} ($A_1$): Swap takes two PRBs $j_1$ and $j_2$ from groups $i_1$ and $i_2$ respectively and assigns $j_1$ to $i_2$ and $j_2$ to $i_1$.
	  Only allocation to the groups $i_1$ and $i_2$ is changed through the swapping action. Mathematically, $s_{d'}$ is obtained from $s_d$ using swap if:
	  \begin{enumerate}
	   \item $j_1 \in \overline{V}_{i_1d}$ and $j_2 \in \overline{V}_{i_2d}$,
	   \item $\overline{V}_{id'} = \overline{V}_{id}$ for all $i \neq i_1,i_2$ and
	   \item $\overline{V}_{i_1d'} = (\overline{V}_{i_1d} \setminus \{j_1\}) \cup \{j_2\}$, $\overline{V}_{i_2d'} = (\overline{V}_{i_2d} \setminus \{j_2\}) \cup \{j_1\}$.
	  \end{enumerate}

	  $\bullet$ \textit{Drop} ($A_2$): The drop action takes a PRB $j_1$ from a group $i_1$ ($i_1 \neq 0$) and assigns it to group $G_0$.
	  Here, only allocation of groups $i_1$ and $0$ is changed by dropping the PRB $j_1$. Mathematically, $s_{d'}$ is obtained from $s_d$ using drop if:
	  \begin{enumerate}
	   \item $j_1 \in \overline{V}_{i_1d}$,
	   \item $\overline{V}_{id'} = \overline{V}_{id}$ for all $i \neq i_1,0$ and
	   \item $\overline{V}_{i_1d'} = \overline{V}_{i_1d} \setminus \{j_1\}$, $\overline{V}_{0d'} = \overline{V}_{0d} \cup \{j_1\}$.
	  \end{enumerate}
	  
	  $\bullet$ \textit{Add} ($A_3$): The add action takes a PRB $j_1$ from $\overline{V}_{0d}$ and assigns it to a group $i_1 \neq 0$.  
	   Here, only allocation of groups $i_1$ and $0$ is changed by assigning the PRB $j_1$ to group $i_1$. Mathematically, $s_{d'}$ is obtained from $s_d$ using add if:
	  \begin{enumerate}
	   \item $j_1 \in \overline{V}_{0d}$,
	   \item $\overline{V}_{id'} = \overline{V}_{id}$ for all $i \neq i_1,0$ and
	   \item $\overline{V}_{i_1d'} = \overline{V}_{i_1d} \cup \{j_1\}$, $\overline{V}_{0d'} = \overline{V}_{0d} \setminus \{j_1\}$.
	  \end{enumerate}
   Note that the neighboring relation defined here is symmetric in nature. This is proved in the following result.
   
   \begin{lemma} \label{lemma:neighbors}
    The neighboring relation of the DTMC $\{X_n^T\}_{n \geq 1}$ is symmetric. Moreover, if transition from $s_d$ to $s_{d'}$ occurs due to a swap action, then transition from $s_{d'}$ to $s_d$ can also take place using a swap action only. 
    Similarly, if transition to $s_d$ from $s_{d'}$ occurs due to add (drop, respectively), the transition from $s_{d'}$ to $s_d$ can only result from drop (add, respectively).
   \end{lemma}

   \begin{proof}
    To prove the required result, we need to show that if a state $s_{d'}$ is a neighbor of the state $s_d$, then, $s_d$ is also a neighbor of $s_{d'}$. Since neighbors are defined using three different actions, we consider the following cases separately:
    
    $\bullet$ \textbf{Swap:} Consider that $s_{d'}$ is obtained from $s_d$ by swapping PRBs $j_1$ and $j_2$ belonging to groups $i_1$ and $i_2$ respectively. Then, from the definition of the swap action,
    $\overline{V}_{id'} = \overline{V}_{id}$ for all $i \neq i_1,i_2$, $\overline{V}_{i_1d'} = (\overline{V}_{i_1d} \setminus \{j_1\}) \cup \{j_2\}$ and $\overline{V}_{i_2d'} = (\overline{V}_{i_2d} \setminus \{j_2\}) \cup \{j_1\}$.
    Now, let us see if state $s_d$ can be obtained from $s_{d'}$. Say PRBs $j_1$ and $j_2$ are picked for swapping in $s_{d'}$. Note that in $s_{d'}$, $j_1 \in \overline{V}_{i_2d'}$ and $j_2 \in \overline{V}_{i_1d'}$. For the resulting state $s_{d''}$, we have:
    \begin{eqnarray*}
     \overline{V}_{id''} &=& \overline{V}_{id'} = \overline{V}_{id}, \forall \ i \neq i_1,i_2, \\
     \overline{V}_{i_1d''} &=& (\overline{V}_{i_1d'} \setminus \{j_2\}) \cup \{j_1\} = \overline{V}_{i_1d}, \\
      \overline{V}_{i_2d''} &=& (\overline{V}_{i_2d'} \setminus \{j_1\}) \cup \{j_2\} = \overline{V}_{i_2d}.    
    \end{eqnarray*}
    
    Therefore, $\overline{V}_{id''} = \overline{V}_{id}$ for all $i$ which implies that $s_{d''} \equiv s_d$. So, $s_d$ is also a neighbor of $s_{d'}$ and can be obtained from $s_{d'}$ using a swap action only.
    
    $\bullet$ \textbf{Add:} Consider that $s_{d'}$ is obtained from $s_d$ by adding PRB $j_1$ to group $i_1$. Then, from the definition of the add action, 
    $\overline{V}_{id'} = \overline{V}_{id}$ for all $i \neq i_1,0$, $\overline{V}_{i_1d'} = \overline{V}_{i_1d} \cup \{j_1\}$ and $\overline{V}_{0d'} = \overline{V}_{0d} \setminus \{j_1\}$.
    Now, let us see if state $s_d$ can be obtained from $s_{d'}$. Say PRB $j_1$ is picked for a drop action in $s_{d'}$. Note that in $s_{d'}$, $j_1 \in \overline{V}_{i_1d'}$. For the resulting state $s_{d''}$, we have:
    \begin{eqnarray*}
       \overline{V}_{id''} &=& \overline{V}_{id'}, \forall \ i \neq i_1,0, \\
       \overline{V}_{i_1d''} &=& \overline{V}_{i_1d'} \setminus \{j_1\} = \overline{V}_{i_1d}, \\
       \overline{V}_{0d''} &=& \overline{V}_{0d'} \cup \{j_1\} = \overline{V}_{0d}.
    \end{eqnarray*}
 Therefore, $\overline{V}_{id''} = \overline{V}_{id}$ for all $i$ which implies that $s_{d''} \equiv s_d$. So, $s_d$ is also a neighbor of $s_{d'}$ and can be obtained from $s_{d'}$ using a drop action only.
 
 $\bullet$ \textbf{Drop:} The proof for the drop action is very similar to that for the add action. It can be shown in the same manner that if $s_{d'}$ is obtained from $s_d$ using a drop action, $s_d$ can be obtained from $s_{d'}$ using an add action and so
 $s_d$ is also a neighbor of $s_{d'}$.
 \end{proof}
In the next section, we define the TPM.
  
 \subsubsection{Transition Probability Matrix}
Let $p_{dd'}$ denote the probability that the DTMC transitions to $s_{d'}$ in the next step from the current state $s_d$. The transition happens in two steps. $1$) In state $s_d$, we first randomly choose one of the three actions $A_1,A_2$ or $A_3$ and then 
randomly choose a neighboring state $s_{d_p}$ that can be obtained from $s_d$ by performing the chosen action. The state $s_{d_p}$ is referred to as the proposed next state. $2$) Based on the reward values $E(s_d)$ and $E(s_{d_p})$, the proposed transition from
$s_d$ to $s_{d_p}$ is either accepted, i.e. $s_{d'} = s_{d_p}$ or rejected, i.e. $s_{d'} = s_d$. Next, we discuss these steps in detail.\\
$\bullet$ Step 1: In this step, one of the three actions is picked. Probability of picking every action is different. Action $A_1$ is picked with probability (w.p.) $\beta_{dA_1} = \frac{1}{3}$, $A_2$ is picked w.p.
$\beta_{dA_2} = \frac{2}{3} \times \frac{N-|\overline{V}_{0d}|}{L(|\overline{V}_{0d}|+1) + (N-(|\overline{V}_{0d}|+1))}$
and $A_3$ is picked w.p. $\beta_{dA_3} = \frac{2}{3} \times \frac{L|\overline{V}_{0d}|}{L|\overline{V}_{0d}| + (N-|\overline{V}_{0d}|)}$. 
With the remaining probability, the state of the DTMC remains unchanged. $A_3$ corresponds to the add action and so, is chosen with a probability directly proportional to the
number of unused PRBs and the number of multicast groups. Therefore, for greater number of groups and unused PRBs, the algorithm is more likely to choose the add action. Similarly, for greater number of used PRBs, the algorithm is more likely to
choose the drop action. \par
Now we explain how one of the neighboring states is chosen for potential transition given the chosen action.
If the chosen action is $A_1$, the two PRBs to be swapped, $j_1$ and $j_2$ are chosen uniformly at random from $[N]$. The swap of $j_1$ and $j_2$ is then performed as discussed in Section~\ref{neighboring_states}.
For $A_2$, the PRB to be dropped, $j_1$ is picked uniformly at random from $[N]\setminus\overline{V}_{0d}$ and dropped as discussed in Section~\ref{neighboring_states}. Similarly for $A_3$,
a group $i_1$ is picked uniformly at random from $[L]$ and a PRB to be added to it, $j_1$ is chosen uniformly at random from $\overline{V}_{0d}$. The addition of $j_1$ to $i_1$ is then done as discussed in Section~\ref{neighboring_states}.
Since different actions lead to different sets of potential neighboring states, we will use $s_{d_{A_i}}$ to denote a state that can be obtained from $s_d$ by performing action $A_i, i \in \{1,2,3\}$. 
In the next step, we discuss how the transition probabilities are finally determined.\\
$\bullet$ Step 2: Let $s_{d'}$ denote the proposed state for transition. If $s_{d'}$ has reward greater than or equal to that of $s_d$, the DTMC transitions to $s_{d'}$. Otherwise, the transition to $s_{d'}$ 
takes place w.p. $e^{(-(E(s_d)-E(s_{d'}))/T)}$. Thus, the probability that the DTMC will transition to the proposed state $s_{d'}$ is $\alpha_{dd'} = \min\left(1,e^{(-(E(s_d)-E(s_{d'}))/T)}\right)$.
Here, $T$ is a parameter commonly known as `temperature'. Unless stated otherwise, we assume $T>0$ is fixed and we denote the corresponding time homogeneous DTMC by $\{X_n^T\}_{n \geq 1}$. \par   
   
 Let $s_{d_{A_1}}, s_{d_{A_2}}$ and $s_{d_{A_3}}$ denote the states resulting from $s_d$ due to $A_1, A_2$ and $A_3$ respectively.
 Then the corresponding transition probabilities take the following form :
  \begin{eqnarray}
 \label{eqn:tpm_a1}  &p_{dd_{A_1}} =  \frac{1}{N(N-1)} \beta_{dd_{A_1}} \alpha_{dd_{A_1}}, \\
\label{eqn:tpm_a2} &p_{dd_{A_2}} = \frac{1}{N-|\overline{V}_{0d}|} \beta_{dd_{A_2}} \alpha_{dd_{A_2}}, \\
\label{eqn:tpm_a3} &p_{dd_{A_3}} = \frac{1}{|\overline{V}_{0d}|} \frac{1}{L} \beta_{dd_{A_3}} \alpha_{dd_{A_3}}, \\
\label{eqn:tpm_4} &p_{dd'} = 0, \text{if} \ s_{d'} \ \text{is not a neighbor of} \ s_d.
   \end{eqnarray}
Note that (\ref{eqn:tpm_a1}), (\ref{eqn:tpm_a2}), (\ref{eqn:tpm_a3}) and (\ref{eqn:tpm_4}) completely describe the TPM.
In the randomized scheme here, we aim to simulate this DTMC with these transition probabilities.
The steps involved in the randomized scheme are presented in the form of a pseudo-code in Algorithm~\ref{algo:SA}. Note that the TPM of the DTMC is not being stored in this algorithm and the transition probabilities defined above can be determined
in polynomial time. Thus, the TPM satisfies all the conditions stated above for computational feasibility of the algorithm. In the next result, we prove certain important properties of the DTMC.

 \begin{algorithm}
	\KwIn{Rates $r_{ij} \forall \ i\in[L]$ and $j\in[N]$, $max\_iter$ = $10^5$}
	Initialize: $s_0$, initial random allocation state\\
	$s_d \leftarrow s_0$\\
	
	\For{$d' = 1:max\_iter$}{\label{for_start}
	$s_{d'} \leftarrow s_d$\\
	$T \leftarrow \frac{1}{\log(k)}$\\
		Pick action $A_1, A_2$ or $A_3$ w.p. $\beta_{dA_1}$, $\beta_{dA_2}$ and $\beta_{dA_3}$ respectively \\

		\uIf{{\rm action}=$A_1$} 
		{Pick any two PRBs, $j_1,j_2 \in [N]$. Say, $j_1 \in \overline{V}_{i_1d'} \ \& \ j_2 \in \overline{V}_{i_2d'}$\\
		$\overline{V}_{i_1d'} = \overline{V}_{i_1d'} \setminus \{j_1\} \cup \{j_2\}$, $\overline{V}_{i_2d'} = \overline{V}_{i_2d'} \setminus \{j_2\} \cup \{j_1\}$\\}
		\uElseIf{{\rm action}=$A_2$}{
		Pick a PRB, $j \in \{\overline{V}_{1d'}, \ldots \overline{V}_{Ld'}\}$. Say, $j \in \overline{V}_{id'}$\\
		$\overline{V}_{id'} = \overline{V}_{id'} \setminus \{j\}$, $\overline{V}_{0d'} = \overline{V}_{0d'} \cup \{j\}$\\}
		\Else{ Pick any $j \in \overline{V}_{0d'}$ and any $i \in \{1,2, \ldots ,L\}$\\
		$\overline{V}_{id'} = \overline{V}_{id'} \cup \{j\}$, $\overline{V}_{0d'} = \overline{V}_{0d'} \setminus \{j\}$
		} 
		$s_d \leftarrow s_{d'}$, if $E(s_{d'}) \geq E(s_d)$ \\
		$s_d \leftarrow s_{d'}$ w.p. $e^{(-(E(s_d)-E(s_{d'}))/T)}$, otherwise\\
		    	}
	$s_d$ is the optimal resource allocation
	\caption{Algorithm for the Randomized Scheme}
	\label{algo:SA}
\end{algorithm}

 \begin{lemma} \label{dtmc_features}
  The constructed DTMC $\{X_n^T\}_{n \geq 1}$ is finite, aperiodic and irreducible for every $T \in (0,\infty)$.
 \end{lemma}
\begin{proof}
The DTMC is finite because the total number of possible resource allocation states is $(L+1)^N$. The DTMC has self loops as there is a positive probability of remaining in
the same state. Hence, the DTMC is aperiodic. The DTMC can transition from any state $s_d$ to any other state $s_{d'}$ by first dropping all the used PRBs into $G_0$ by 
choosing the drop action repeatedly. Then, the PRBs can be added one by one according to the assignment in state 
$s_{d'}$ by choosing the add action repeatedly. Therefore, the DTMC is irreducible.
\end{proof}
Having established that the DTMC is finite, aperiodic and irreducible, it is guaranteed to have a unique steady state distribution. In the following result, we determine this steady state distribution.

\begin{theorem}
 For any fixed $T>0$, the steady state distribution of the DTMC $\{X_n^T\}_{n \geq 1}$ is given by
 \begin{equation*}
  \pi_d^T = \frac{e^{E(s_d)/T}}{\sum_{s_d} e^{E(s_d)/T}} \forall \ s_d \in \chi.
 \end{equation*}
\end{theorem}

\begin{proof}
To prove the required, we show that the transition probabilities satisfy $\pi_d^Tp_{dd'} = \pi_{d'}^Tp_{d'd}$ for every $s_d, s_{d'}$. This will imply that the DTMC is reversible and has steady state distribution 
$\pi_d^T = \frac{e^{E(s_d)/T}}{\sum_{s_d} e^{E(s_d)/T}}, \forall \ s_d \in \chi$. \par
Suppose $s_d$ and $s_{d'}$ are not neighboring states, then $p_{dd'} =  p_{d'd} = 0$. Hence, the required follows trivially. Thus, it suffices to consider the case when $s_d$ and $s_{d'}$ are neighbors.
If $s_d$ and $s_{d'}$ are neighbors, there are three possibilities, that $s_{d'}$ is obtained from $s_d$ by $1$) swap action, $2$) drop action or $3$) add action. We consider each case separately:\\
\noindent
 $\bullet$ \textbf{Swap:} If the transition from $s_d$ to $s_{d'}$ occurs due to a swap action, then $p_{dd'}$ and $p_{d'd}$ take the form given by~(\ref{eqn:tpm_a1}). For $ E(s_d) \geq E(s_{d'})$ we have:
 \begin{align*}
 \lefteqn{ \frac{e^{E(s_d)/T}}{\sum_{d \in \chi} e^{E(s_d)/T}} \frac{1}{3} \frac{1}{N(N-1)} e^{-(E(s_d)-E(s_{d'}))/T}} \\
  &= \frac{e^{E(s_{d'})/T}}{\sum_{d \in \chi} e^{E(s_d)/T}} \frac{1}{3} \frac{1}{N(N-1)}, \\ 
 \end{align*}
which is true. Therefore, the given $\pi_d^T$ satisfies $\pi_d^Tp_{dd'} = \pi_{d'}^Tp_{d'd}$ for the swap action.
This can be similarly shown for $E(s_d) < E(s_{d'})$ as well.\\
$\bullet$ \textbf{Add:} If the transition from $s_d$ to $s_{d'}$ occurs due to an add action, $p_{dd'}$ and $p_{d'd}$ will be given by~(\ref{eqn:tpm_a3}) and~(\ref{eqn:tpm_a2}) respectively.
For $E(s_d) \geq E(s_{d'})$, we have:
\begin{align}
  \lefteqn{\frac{2\pi_d^T}{3\left(L|\overline{V}_{0d}| + (N-|\overline{V}_{0d}|)\right)} e^{-(E(s_d)-E(s_{d'}))/T} } \nonumber \\
  \label{eqn:add_tpm} &= \frac{2\pi_{d'}^T}{3\left(L(|\overline{V}_{0d'}|+1) + (N-(|\overline{V}_{0d'}|+1))\right)}.
\end{align}
Since $s_{d'}$ is obtained from $s_d$ using an add action, $|\overline{V}_{0d}| = |\overline{V}_{0d'}|+1$ which means that $L|\overline{V}_{0d}| + (N-|\overline{V}_{0d}|) = L(|\overline{V}_{0d'}|+1) + (N-(|\overline{V}_{0d'}|+1))$ in (\ref{eqn:add_tpm}) above.
So, (\ref{eqn:add_tpm}) becomes:
\begin{eqnarray*}
  \pi_d^T e^{-(E(s_d)-E(s_{d'}))/T} &=& \pi_{d'}^T, \\
  \implies \frac{e^{E(s_d)/T}}{\sum_d e^{E(s_d)/T}} e^{-(E(s_d)-E(s_{d'}))/T} &=& \frac{e^{E(s_{d'})/T}}{\sum_d e^{E(s_d)/T}},
\end{eqnarray*}
which is true. Therefore, the given $\pi_d^T$ satisfies $\pi_d^Tp_{dd'} = \pi_{d'}^Tp_{d'd}$ for the add action.
This can be similarly shown for $E(s_d) < E(s_{d'})$ as well.\\
$\bullet$ \textbf{Drop:} If the transition from $s_d$ to $s_{d'}$ occurs due to a drop action, $p_{dd'}$ and $p_{d'd}$ will be given by~(\ref{eqn:tpm_a2})  and~(\ref{eqn:tpm_a3}) respectively.
Also, in this case, $|\overline{V}_{0d'}| = |\overline{V}_{0d}|+1$. Following the same steps as for the add action, it can be shown that the given $\pi_d^T$ satisfies $\pi_d^Tp_{dd'} = \pi_{d'}^Tp_{d'd}$ for the drop action as well. \par

Therefore, we conclude that the steady state distribution of the DTMC $\{X_n^T\}_{n \geq 1}$ is given by $\pi_d^T = \frac{e^{E(s_d)/T}}{\sum_{s_d} e^{E(s_d)/T}}$ for every $s_d \in \chi$.
\end{proof}


Now that we have the steady state distribution of the DTMC for a fixed value of $T > 0$, it can be shown that, as $T$ goes to $0$, the steady state distribution $\pi_d = \lim_{T \rightarrow 0} \pi_d^T$ takes the following form:
\begin{equation*}
 \pi_d = \begin{cases}
          {1/|\arg\max_d E(s_d)|}, & \forall \ d \in \arg\max_d E(s_d),\\
	{0}, &  \text{otherwise}.
         \end{cases}
\end{equation*}
Thus, $\pi_d$ is a uniform distribution over the optimal resource allocation states. \par

We mentioned the parameter $T$ above, while discussing the TPM. Now, we elaborate it's significance in more detail. SA involves an exploration versus exploitation trade-off. 
It achieves a balance between exploration and exploitation through the temperature parameter $T$. $T$ is kept very high in the beginning so that the algorithm can explore a large number of states quickly.
As the time index increases, $T$ goes on decreasing and so does the likelihood of transitioning to lower reward states. The most widely used function for temperature is $1/\log(n)$~\cite{bh}, $n$ being the time index.
This form of $T$ ensures that the algorithm escapes local optima faster and ends up at the global optimum as $T$ goes to $0$.
Specifically, by varying $T$, we can achieve the required $\lim_{n \rightarrow \infty} \mathbb{P}(E(X_n) = E^\star) = 1$.
In the next section, we compare the results of the RS with the optimal solution obtained by solving the BLP \Bstar for small input sizes. 

    \subsection{Performance comparison of the RS and the BLP}
    The optimal resource allocation can be obtained by solving the BLP \Bstar \ from Section~\ref{sec:sys_model}. BLPs, as mentioned before, are inherently hard to solve. They can however be solved for small input sizes. 
    Using the computing power at our disposal (Intel i7, $2.90$ GHz quad-core processor with $16$ GB RAM), we were able to obtain a solution of \Bstar \ for an input size of up to $4$ groups.
    Note that the search space scales as $(L+1)^N$ where $L$ is the number of groups and $N$ is the number of PRBs in a sub-frame. So, even for $4$ groups and $100$ PRBs, the search space consists of $5^{100}$ states which is why the BLP fails to give a 
    solution for more than $4$ groups. The outputs of the BLP and the RS for up to $4$ groups, 
    averaged over $100$ different channel conditions are tabulated in Table~\ref{bilp}. As we can see, the output of the RS is very close (difference in number of PRBs saved $< 5.5\%$) to the optimal obtained by solving the BLP.
    
    
    \begin{center}
    \begin{table}
\captionof{table}{Performance comparison of RS and BLP} \label{bilp} 
 \begin{center}
\begin{tabular}{ | m{2 cm} | m{1.2 cm}| m{1.2 cm}| m{1.2cm}|} 
 \hline
 \textbf{No. of groups} & \textbf{RS} & \textbf{BLP} & \textbf{\%~Error}\\
 \hline
 \hline
 $2$ & $92$ & $96$ & $4.16$ \\ 
    \hline
  $3$ & $90$ & $94$ & $4.25$ \\ 
  \hline
   $4$ & $86$ & $91$ & $5.49$ \\
 \hline 
 \end{tabular}
\end{center}
     \end{table}
      \end{center}
     
    The randomized scheme works iteratively to obtain an optimal solution and so, it is not guaranteed to converge within the sub-frame duration of $1$ ms. We require resource allocation schemes that can output a near optimal
    solution (if not optimal) every sub-frame. We now present two heuristic schemes that give us a reasonably good performance. We compare the output of one of the proposed schemes with the output of the RS and show that it gives a solution
    very close to the optimum and takes significantly less time to run than the RS.
    
    \section{Heuristic Schemes for Resource Allocation} \label{sec:RASh}
    In this section, we propose two heuristic schemes for allocating PRBs to multicast groups. The first scheme allocates PRBs greedily and the second one makes use of Linear Programming (LP) relaxation.
    Allocation of resources in LTE is done in every sub-frame. So, for brevity, we fix a sub-frame $t$ and omit it from notations in this section. Grouping strategy $\Delta$ impacts resource allocation via $r_{ij}^\Delta$, which is the rate 
achievable by group $i$ in PRB $j$. Our aim is to propose resource allocation for any given $\Delta$. So, we fix $\Delta$ and omit it from the notations as well. 

\subsection{Greedy Allocation} \label{greedyalgo}
\begin{algorithm}
	\KwIn{Rates $r_{ij}$ for all $i\in[L]$ and $j\in[N]$}
	Initialize: ${\cal N}= [N]$, ${\cal L} =[L]$ and $x_{ij} = 0$ for every $i,j$\\
	\While{${\cal N} \cap {\cal L} \not= \phi$}{\label{ln:first_while}
		Assign $(i^\star,j^\star) = \arg\max_{(i,j)\in {\cal N}\times{\cal L} }r_{ij}$ \\
		$x_{i^\star j^\star} \leftarrow 1$, 
		${\cal N} \leftarrow {\cal N} \setminus \{j^\star\}$\\ 
		\If{$\sum_{j\in [N]} x_{i^\star j} r_{i^\star j} \ge R$}{${\cal L} \leftarrow {\cal L} \setminus \{i^\star\}$}
	}
	\caption{Greedy Resource Allocation Scheme}
	\label{fig:greedy}
\end{algorithm}
The pseudo code for this scheme is given in Algorithm~\ref{fig:greedy}. Here, ${\cal N}$ and ${\cal L}$ denote the unallocated PRBs and the groups whose rate requirements are
not yet satisfied, respectively. These quantities are updated every iteration and are monotone non-increasing. The algorithm terminates when either of the two sets becomes empty. In each iteration, the algorithm determines 
indices $i^\star$ and $j^\star$ from ${\cal L}$ and ${\cal N}$, respectively, that correspond to the maximum $r_{ij}$. PRB $j^\star$ is allotted to group $i^\star$ and is removed from ${\cal N}$.
Also, if the total sum rate on all the allotted PRBs to $i^\star$ is greater than or equal to the requirement $R$,
then $i^\star$ is also removed from ${\cal L}$.  Next iteration starts with the new values of ${\cal N}$ and 
${\cal L}$. Note that ${\cal N}$ is monotone decreasing, thus, the algorithm terminates in at most $N$ 
iterations. At the termination, if only ${\cal N} = \phi$ and ${\cal L}$ is non-empty, then the greedy resource allocation scheme fails to output a feasible resource allocation, else variables $x_{ij}$'s yield the required resource allocation.
The resource allocation thus obtained is inherently fair as the algorithm provides the minimum required rate $R$ to all the UEs.

\subsection{LP-relaxation Based Allocation}
Recall that the optimal resource allocation can be obtained as a solution to the BLP \Bstar. BLPs are inherently hard to solve and cannot be solved in polynomial time except for very small input sizes.
One standard approach is to consider LP-relaxation of the BLP i.e., relax the binary variables (in our case, $x_{ij}$'s) to take values in the interval $[0,1]$. The resulting LP can be solved in
polynomial time. Let $\tilde{x}_{ij}$ for all $i,j$ denote the optimal solution of the relaxed LP. Now, $\tilde{x}_{ij}$'s are real numbers and we need to convert them to binary values
without violating the constraints of \Bstar. To do so, we use a greedy algorithm (Algorithm~\ref{algo:LP}) similar to the one used in Section~\ref{greedyalgo} above.
In each iteration, a PRB $j$ is assigned to an unsatisfied group $i$ if it has the largest value of $\tilde{x}_{ij}$ for that PRB. This is intuitive, as a higher value of $\tilde{x}_{ij}$ indicates that group $i$ was assigned a larger share of PRB $j$
by the LP. Note that the resource allocation obtained using this LP-relaxation based scheme is inherently fair as the algorithm ensures that the minimum required rate $R$ is provided to all the UEs.

\begin{algorithm}
	\KwIn{$\tilde{x}_{ij}$ for all $i\in[L]$ and $j\in[N]$}
	Initialize: ${\cal N}= [N]$, ${\cal L} =[L]$ and $x_{ij} = 0$ for every $i,j$\\
	\While{${\cal N} \cap {\cal L} \not= \phi$}{\label{ln:while_LP}
		Assign $(i^\star,j^\star) = \arg\max_{(i,j)\in {\cal N}\times{\cal L} }\tilde{x}_{ij}$ \\
		$x_{i^\star j^\star} \leftarrow 1$, 
		${\cal N} \leftarrow {\cal N} \setminus \{j^\star\}$\\ 
		\If{$\sum_{j\in [N]} x_{i^\star j} r_{i^\star j} \ge R$}{${\cal L} \leftarrow {\cal L} \setminus \{i^\star\}$}
	}
	\caption{Rounding off algorithm for LP-relaxation}
	\label{algo:LP}
\end{algorithm}


    
\subsubsection{Performance Comparison of RS and LP-relaxation}
In order to compare the performance of the LP-relaxation based allocation to that of the RS, we simulate an LTE cell with all the multicast UEs requiring the same content from the eNB. PRBs are allocated to the UEs using the RS as well as the LP-relaxation scheme.
We gradually increase the number of UEs in the cell starting from $10$ UEs and go up to $100$, adding $10$ UEs at a time. For each of the resulting $10$ scenarios, the PRB allocation is done for $100$ different fading variations using both the schemes.
The average number of PRBs saved is used as a measure for performance comparison. The results of the simulations are plotted in Fig.~\ref{SA}. Each point in the curves has been obtained by averaging over $100$ different channel gain variations.
Note that all the groups achieved the required rates at all points in the two curves.
Both the algorithms show a similar trend as the number of UEs in the cell increases. Even though the RS saves more PRBs throughout, the ratio of the number of PRBs saved by the RS to the number of PRBs saved by the LP-relaxation scheme is no more than $1.25$.

    
    \begin{center}
    \begin{table}
\captionof{table}{Time taken in seconds to run RS and LP-relaxation based algorithm} \label{time} 
 \begin{center}
\begin{tabular}{ | m{1.5cm} | m{0.7cm}| m{2cm}| m{0.7cm} |} 
 \hline
 \textbf{No. of UEs} & \textbf{RS} & \textbf{LP-relaxation} & \textbf{Ratio} \\
 \hline
 \hline
 $20$ & $0.082$ & $0.015$ & $5.47$ \\ 
\hline
$30$ & $0.086$ & $0.018$ & $4.78$ \\ 
\hline
$40$ & $0.086$ & $0.019$ & $4.53$ \\ 
\hline
$50$ & $0.087$ & $0.020$ & $4.35$ \\ 
\hline
$60$ & $0.089$ & $0.021$ & $4.24$ \\ 
\hline
$70$ & $0.092$ & $0.020$ & $4.6$ \\ 
\hline
$80$ & $0.097$ & $0.017$ & $5.71$ \\ 
\hline
$90$ & $0.097$ & $0.019$ & $5.11$ \\ 
\hline
$100$ & $0.096$ & $0.018$ & $5.33$ \\ 
\hline
 \end{tabular}
 \end{center}
 \end{table} 
\end{center}

    \subsubsection{Time Comparison of RS and LP-relaxation}
    Recall that, in LTE, the allocation of PRBs is done every sub-frame.
Since a sub-frame spans only $1$ ms in time, it is important for whatever resource allocation scheme we employ, to be time efficient as well.
    We now do a brief time comparison of the RS and LP-relaxation schemes. \par
    The RS is an iterative algorithm and cannot be guaranteed to converge within the span of a sub-frame. While simulating the RS in this paper, we perform $10^5$ iterations. However, for the time comparison here, we will first see how 
    the reward of the current state of the RS changes as a function of the number of iterations. Fig.~\ref{iter} illustrates the change in the reward of the current state of the RS as a function of the number of iterations
    for different number of UEs in the cell. We can observe from the figure that the output saturates well before $2000$ iterations in each curve. 
    So, for the sake of time comparison with the LP-relaxation scheme, we consider the time taken by just $2000$ iterations of the RS. Table~\ref{time} illustrates the time taken by the RS and the LP-relaxation scheme for different number of UEs in the cell.
    The time taken is averaged over $200$ different channel gains. We observe that the RS takes about $5$ times more time to run than the LP-relaxation scheme even with just $2000$ iterations. Note that in practice, depending upon the system, we might need
    to run the algorithm for a much larger number of iterations. \par    
    From the performance and time comparisons of the LP-relaxation scheme and the RS, we conclude that the LP-relaxation scheme performs nearly as well as the RS in 
    $5$ times lesser duration than the RS. Thus, the LP-relaxation scheme is a suitable resource
    allocation scheme for practical implementation. In the next section, we present two heuristic schemes for the grouping of UEs for multicast transmission.

    \begin{figure}[h]
    \centering
    \begin{minipage}{0.45\textwidth}
        \centering
       \includegraphics[scale = 0.38]{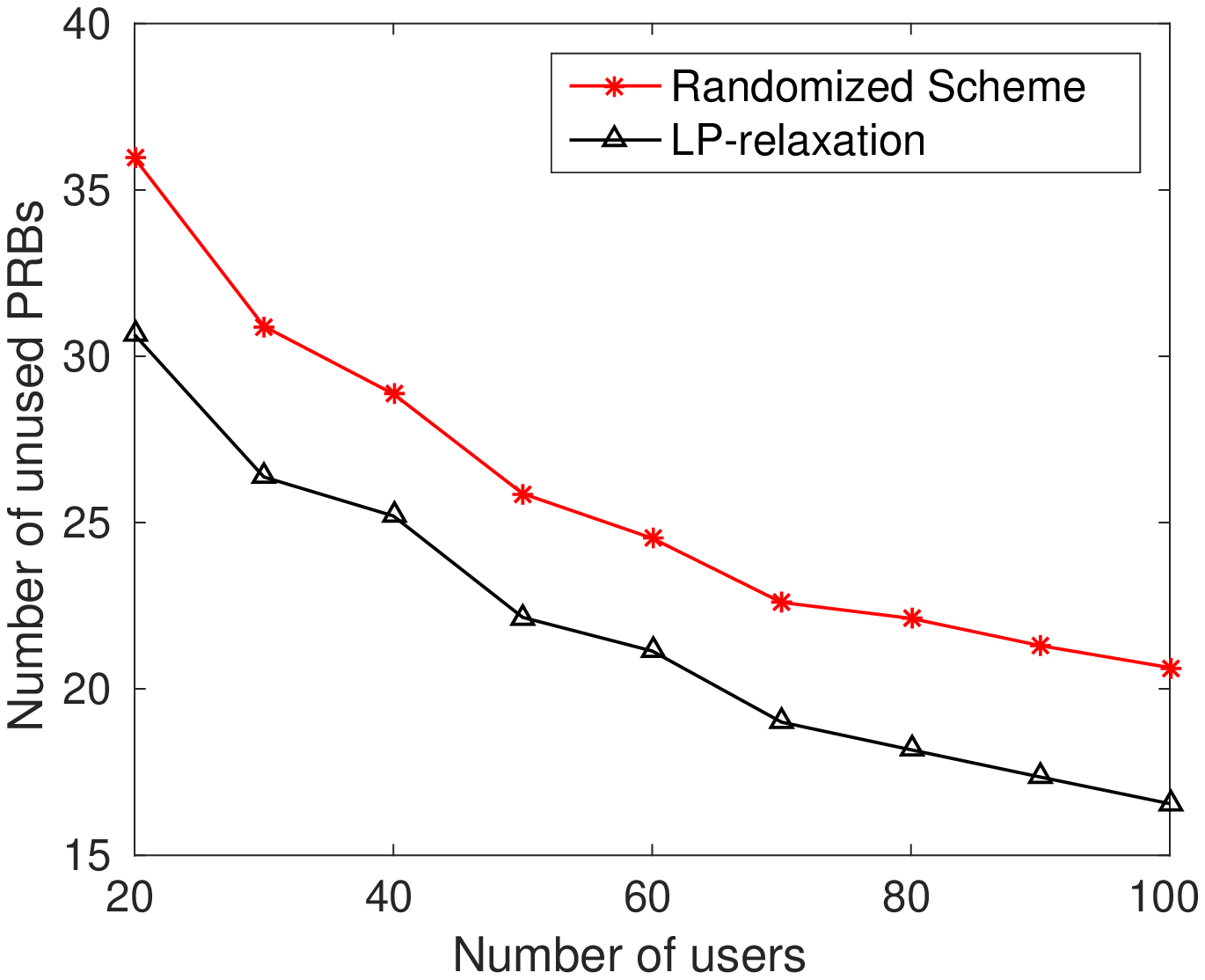} 
        \caption{Variation of the number of unused PRBs : a performance comparison of the LP-relaxation scheme and the RS.}
        \label{SA}
    \end{minipage}\hfill
    \begin{minipage}{0.45\textwidth}
        \centering
       \includegraphics[scale = 0.38]{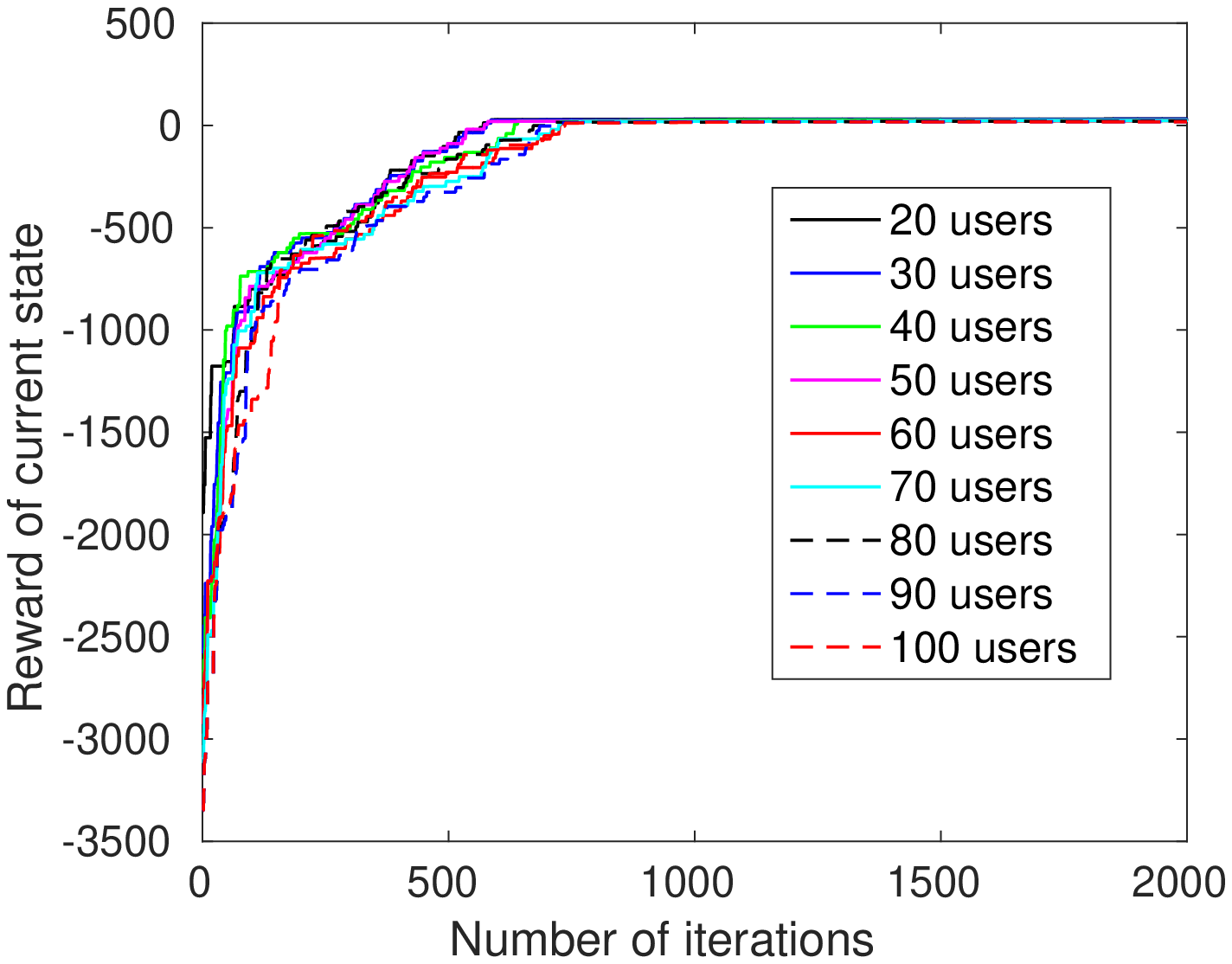}
        \caption{Variation of the reward of the state of RS with the increasing number of iterations.}
        \label{iter}
    \end{minipage}
\end{figure}
    
 \section{Heuristic Schemes for Grouping} \label{sec:grouping}
We have proved that obtaining optimal grouping strategy $\Delta^\star$ for the set of multicast UEs 
that maximizes the performance measure $\overline{S}^{\Delta^\star}$ is 
a NP-hard problem. Indeed, quantifying $\overline{S}^{\Delta}$ for a given grouping strategy
$\Delta$ is a very difficult task as the channel gains and hence the rates vary over time. This is because the
optimal resource allocation in a given sub-frame itself is a NP-hard problem.
However, even if some genie provides us with the value  $\overline{S}^{\Delta}$ for any given $\Delta$,
determining the optimal $\Delta^\star$ is still NP-hard (proved in Section~\ref{NPC2}).
Hence, in this section, we present two heuristic grouping algorithms, namely 
fixed size grouping and CQI based grouping schemes.

\subsection{Fixed Size Grouping}
In fixed size grouping, the group size is fixed at some integral value $k$ and then the UEs are grouped according to their average SNR values. Note that the average
SNR depends only on the slow fading components that include path loss and shadowing. Since we have assumed that the eNB has full CSI, average SNR values of all 
the UEs can be estimated at the eNB. In order to group the UEs, they are sorted in descending order of their average SNR values. The first $k$ UEs then form the first group, the second set of $k$
UEs form the second group and so on. The number of groups thus formed is $\left\lceil \frac{M}{k}\right \rceil$ where $M$ denotes the total number of UEs in the cell.\par

For resource allocation, rates are assigned to the groups on the basis of their instantaneous SNR values. 
The instantaneous SNR of UEs are mapped to CQI values which are further mapped to rates via standard SNR to CQI and CQI to MCS mappings respectively~\cite{seventh}. 
The rate achieved by a group in a particular PRB is equal to the minimum of the rates achievable by it's constituent members.
This is because, transmitting at a rate greater than this may lead to incorrect reception of data at the UE with the worst channel gain and hence the least achievable rate.
Limiting the number of UEs in a group keeps a check on the probability that some UE has much worse channel than that of others in the group for a given PRB. 
Note that, here, we are grouping UEs without considering their exact average SNR values. In the following grouping scheme, the groups are once again formed based on average SNR values but with variable group sizes.

\subsection{CQI based Grouping}
In this scheme, the eNB fixes the number of groups and then the UEs are assigned to various groups based on their average SNR values. 
In 3GPP standards for LTE~\cite{seventh}, a total of 15 CQI values have been defined with 15
indicating the best and $1$ indicating the worst channel. A range of SNR values 
are mapped to a particular CQI value (many to one map). Corresponding to each CQI value, an MCS is also predefined~\cite{seventh}. 
In keeping with the number of CQI values, the maximum possible number of groups under this strategy is fixed to be $15$. 

In LTE, a range of SNR values get mapped to a CQI value~\cite{seventh}. Let the minimum SNR that can be mapped to a CQI value $c$ be denoted by 
$\text{SNR}_{\rm min}(c)$. We define a threshold corresponding to this CQI at such a level that with a large probability (say $0.9$), the instantaneous SNR 
of the weakest UE in the group will stay above or at $\text{SNR}_{\rm min}(c)$. Specifically, a threshold $T(c)$ is defined such that, 
\begin{equation*}
\mathbb{P}\{ \text{SNR} \geq \text{SNR}_{\rm min}(c) | \text{SNR}_{\rm avg} = T(c) \} = 0.9.
\end{equation*}
To compute $T(c)$, we need the distribution of $h_{iu}[t]$. When $H_{iu}[t]$ (the fast-fading component of $h_{iu}[t]$ as defined in Section~\ref{sec:sys_model}) are 
i.i.d exponential with mean $1$, $T(c)$ can be computed using:
\begin{equation*}
T(c) = \frac{\text{SNR}_{\rm min}(c)}{\log(10/9)}.
\end{equation*}

The UEs are classified into groups on the basis of their average SNR values. 
Average SNR values greater than or equal to $T(15)$ are classified as Group 1 
and those with SNR below $T(2)$ are grouped into Group 15. UEs with average 
SNR between $T(14)$ and $T(15)$
are put into Group 2 and so on. Thus, Group 1 (Group 15) corresponds to the 
UEs with the best (worst) channel gains.\par

After the UEs are classified into groups, the rate for a particular group is set at the value corresponding to 
the weakest UE in the group. Once the achievable rate for each group is determined using the 
3GPP mappings~\cite{seventh}, the PRB allocation is done according to 
the resource allocation schemes discussed in the previous section. 
\section{Simulation Results}\label{sec:simulations}

\subsection{Simulation Settings}
Our system comprises of a single LTE cell of radius 375 meters, half of the inter-site distance mentioned in the simulation parameters for macro cell propagation model given in~\cite{sixth}.
We have simulated an LTE cell in MATLAB~\cite{matlab} using the LTE simulator designed in~\cite{mahima}. In order to create LTE specific physical layer conditions, we have created channels using the models recommended by 3GPP in~\cite{sixth}. The SNR to CQI and CQI to rate mapping has been done using the tables specified in the 3GPP documents~\cite{sixth}. \par
 
 An eNB located at the center of the cell that multicasts the MBMS content to all the multicast groups in the cell. The UEs are distributed uniformly at random within the cell and
 are grouped using the fixed and the CQI based grouping schemes proposed in Section~\ref{sec:grouping}. Resource allocation is also done using both the heuristic schemes proposed in Section~\ref{sec:RAS}. We compare the performance of the proposed schemes
 with each other as well as with unicast transmission. For calculating the average SNR, we use shadowing and path loss models as per 3GPP specifications~\cite{sixth}. For instantaneous SNR, we also take Rayleigh fading into account.
The parameters relevant to our simulations are given in Table~\ref{parameters}.
We assume channel gain to be determined by: 1): Path Loss, 2): Shadowing and 3): Multipath due to reflections from the surrounding environment.
The channel gain of each UE may be different for different PRBs. The rate requirement for each UE (R) is taken to be $1$ Mbps.\par

For a given grouping and resource allocation scheme, the performance is affected by two sources of randomness, (1) channel variations around mean on account of fast fading and (2) average channel gain variations on account of node positions.
We evaluate the average performance by averaging over these two sources of randomness. Towards this end, we consider $100$ random UE placements in the cell, and the performance of 
each topology is evaluated over $1000$ sub-frames with different channel gains.\par

In addition to unicast and the proposed grouping schemes, we also consider random grouping in the first set of simulations where each UE is placed in one of the $10$ groups uniformly at random. This is done to evaluate the role of grouping strategy 
in the system performance. The group size for fixed-size grouping is taken to be  $5$. 

\begin{center}
\begin{table}
\captionof{table}{System Simulation parameters \cite{sixth}} \label{parameters} 
\begin{center}
\begin{tabular}{ | m{4 cm} | m{4 cm}|} 
 \hline
 \textbf{Parameters} & \textbf{Values}\\
 \hline
 \hline
 System bandwidth & $20$ MHz \\
 \hline
 eNB cell radius & $375$ m \\ 
 \hline
 Path loss model & L = $128.1 + 37.6 \log10(d)$, $d$ in kilometers\\ 
 \hline
 Lognormal shadowing & Log Normal Fading with $10$ dB standard deviation \\
 \hline
 White noise power density & $-174$ dBm/Hz \\
 \hline
 eNB noise figure & $5$ dB \\
 \hline
 eNB transmit power & $46$ dBm \\
 \hline
  PRB width & $180$ kHz\\
 \hline
 Number of PRBs & $100$ per sub-frame\\
\hline
 \end{tabular}
 \end{center}
 \end{table}
\end{center}

\begin{figure*}[htb]
    \centering
    \begin{subfigure}[t]{0.45\textwidth}
        \centering
\includegraphics[scale = 0.38]{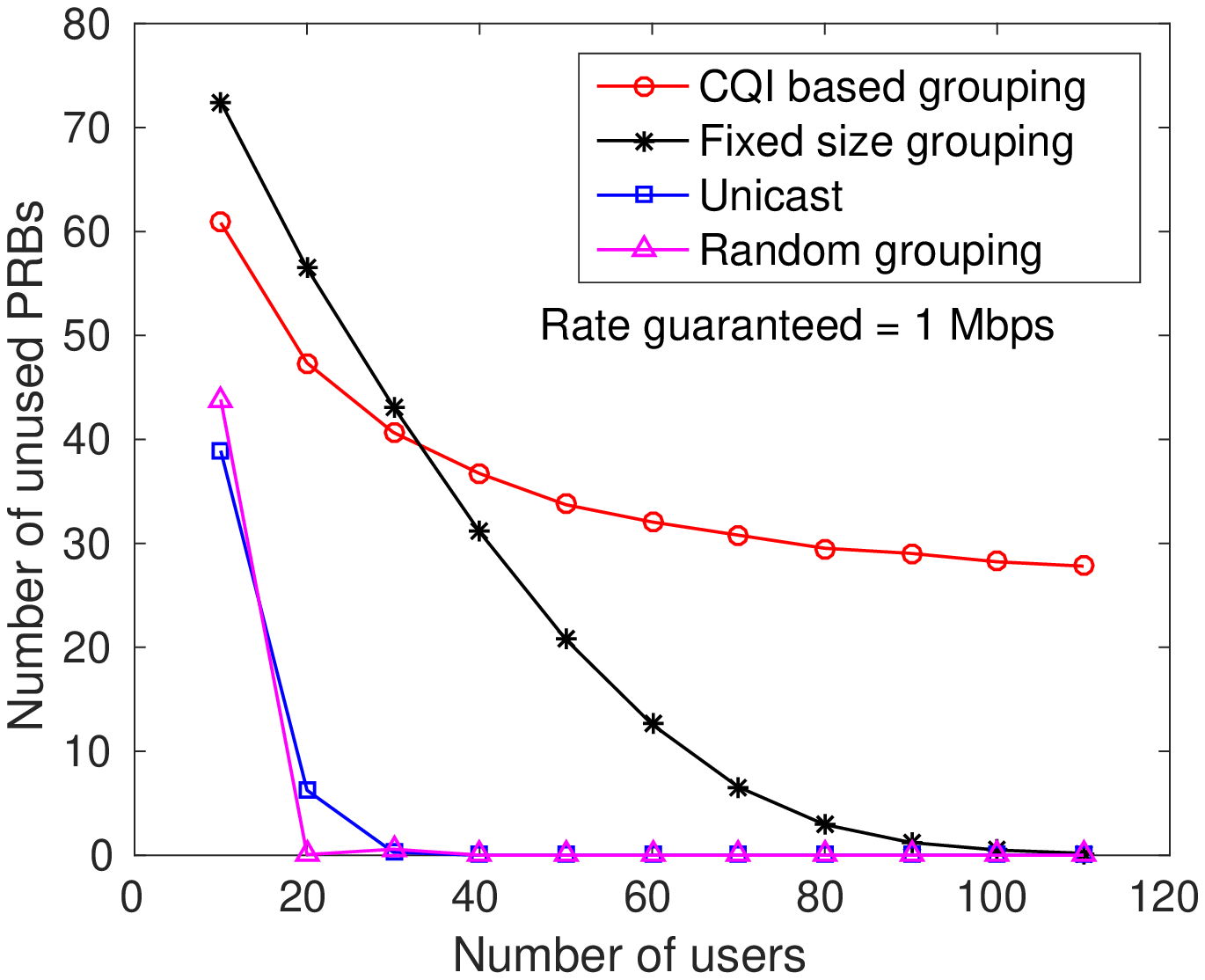} 
\caption{Variation of the number of unused PRBs under Greedy scheme with increasing number of UEs.}
\label{varying_greedy}
    \end{subfigure}%
    ~ 
    \begin{subfigure}[t]{0.45\textwidth}
        \centering
\includegraphics[scale = 0.38]{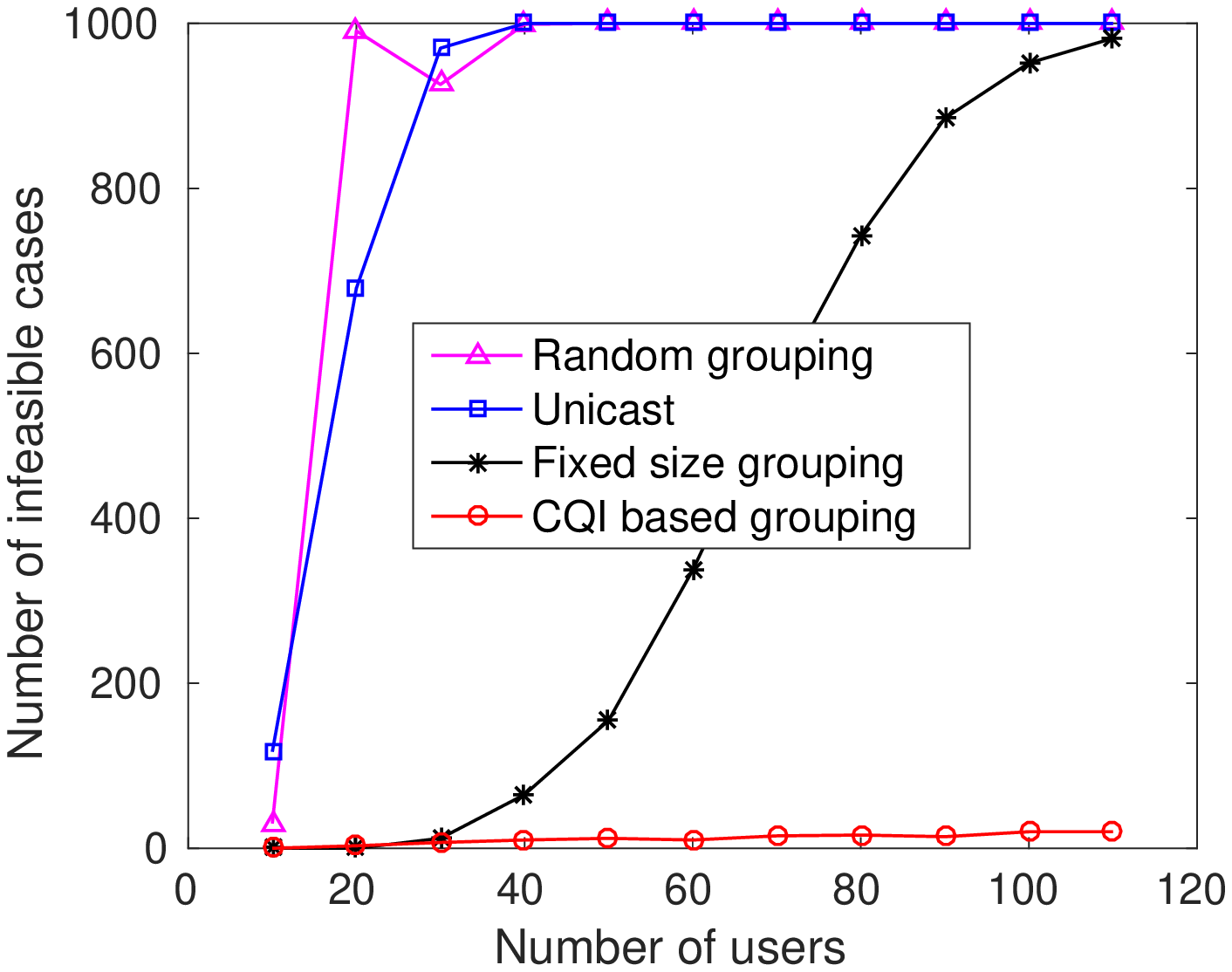}
\caption{Variation of the number of infeasible cases under Greedy scheme with increasing number of UEs.}
\label{varying_infeasible_greedy}
    \end{subfigure}%
~
\caption{Performance evaluation of the Greedy scheme}
\end{figure*}

\begin{figure*}[htb]

    \centering
    \begin{subfigure}[t]{0.45\textwidth}
        \centering
\includegraphics[scale = 0.38]{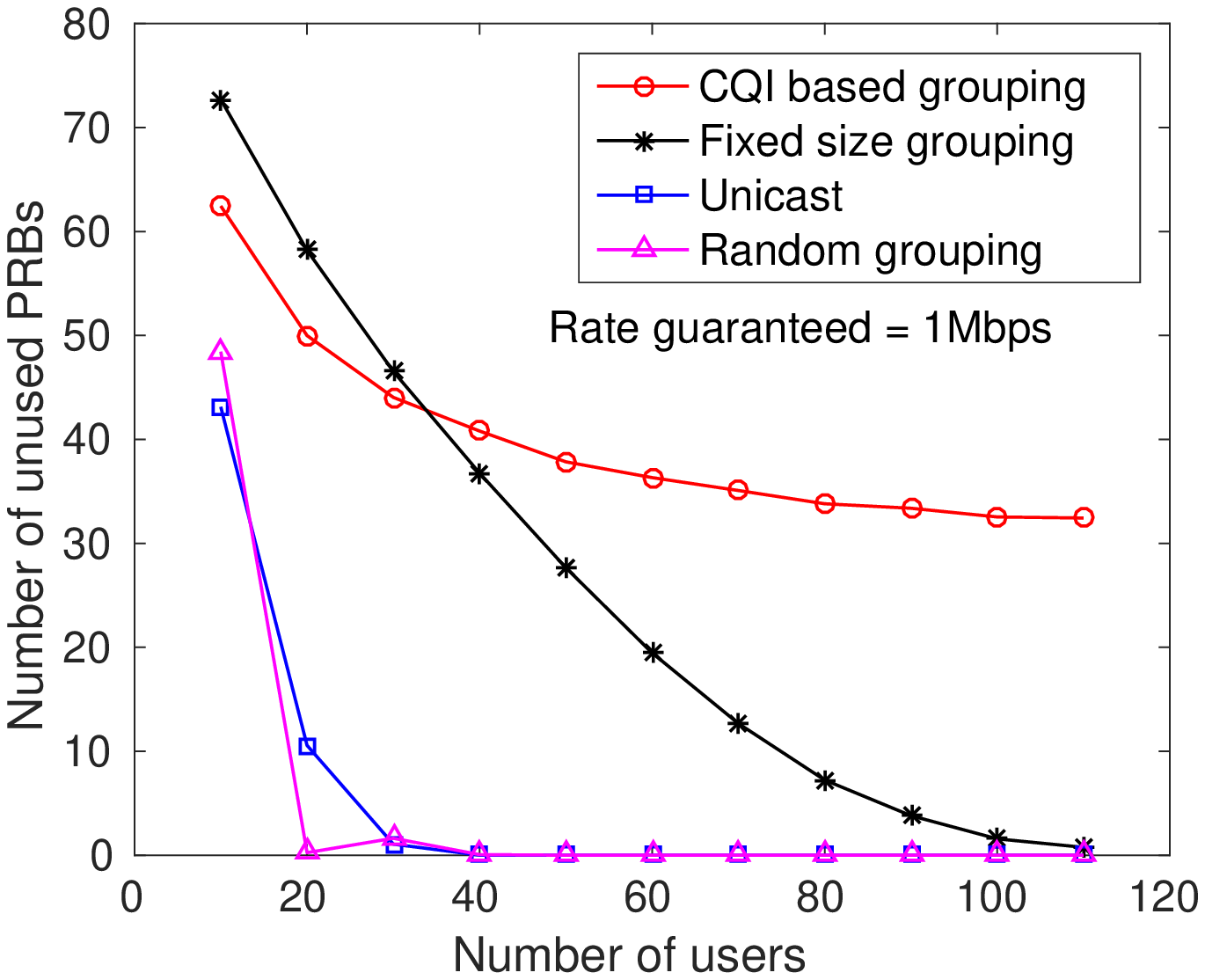} 
\caption{Variation of the number of unused PRBs under LP-relaxation based scheme with increasing number of UEs.}
\label{varying_LP}
\end{subfigure}
    ~ 
    \begin{subfigure}[t]{0.45\textwidth}
        \centering
\includegraphics[scale = 0.38]{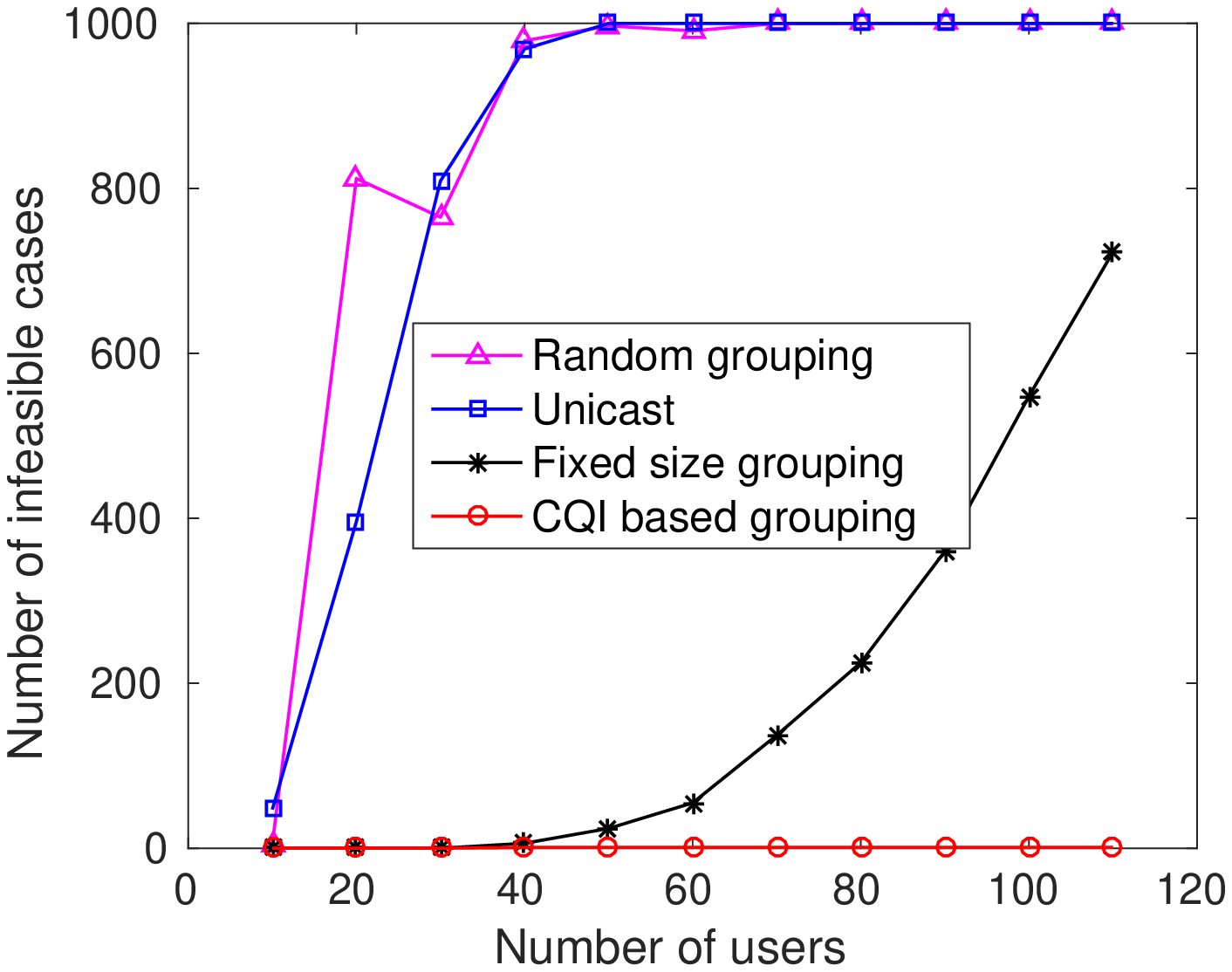}
\caption{Variation of the number of infeasible cases under LP-relaxation based scheme with increasing number of UEs.}
\label{varying_infeasible_LP}
\end{subfigure}
\caption{Performance evaluation of the LP-relaxation based scheme}

\end{figure*}

\begin{figure*}[htb]
    \centering
    \begin{subfigure}[t]{0.45\textwidth}
        \centering
\includegraphics[scale = 0.38]{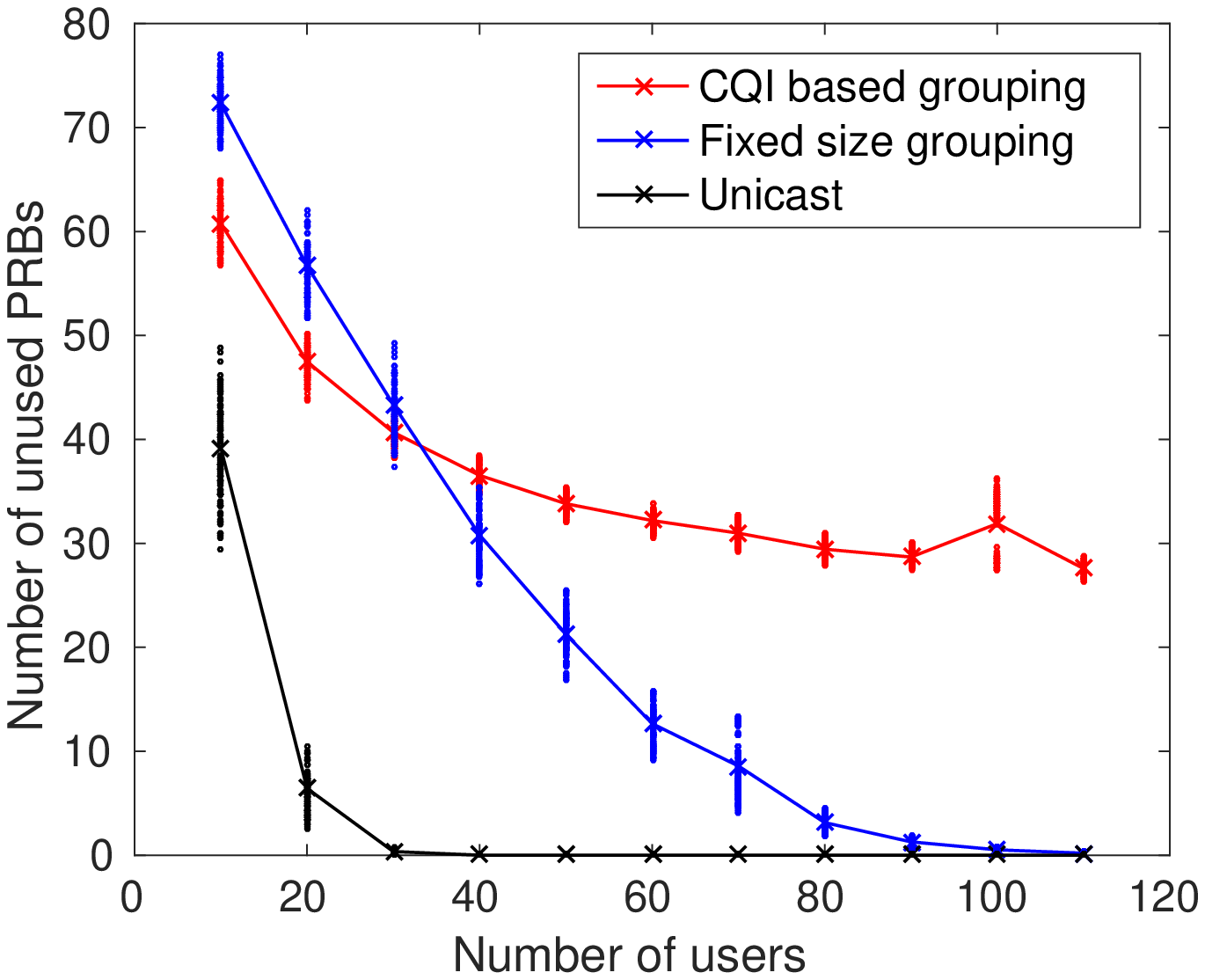} 
\caption{Number of unused PRBs with Greedy scheme for different UE placements.}
\label{greedy_mean}
\end{subfigure}
    ~ 
    \begin{subfigure}[t]{0.45\textwidth}
        \centering
\includegraphics[scale = 0.38]{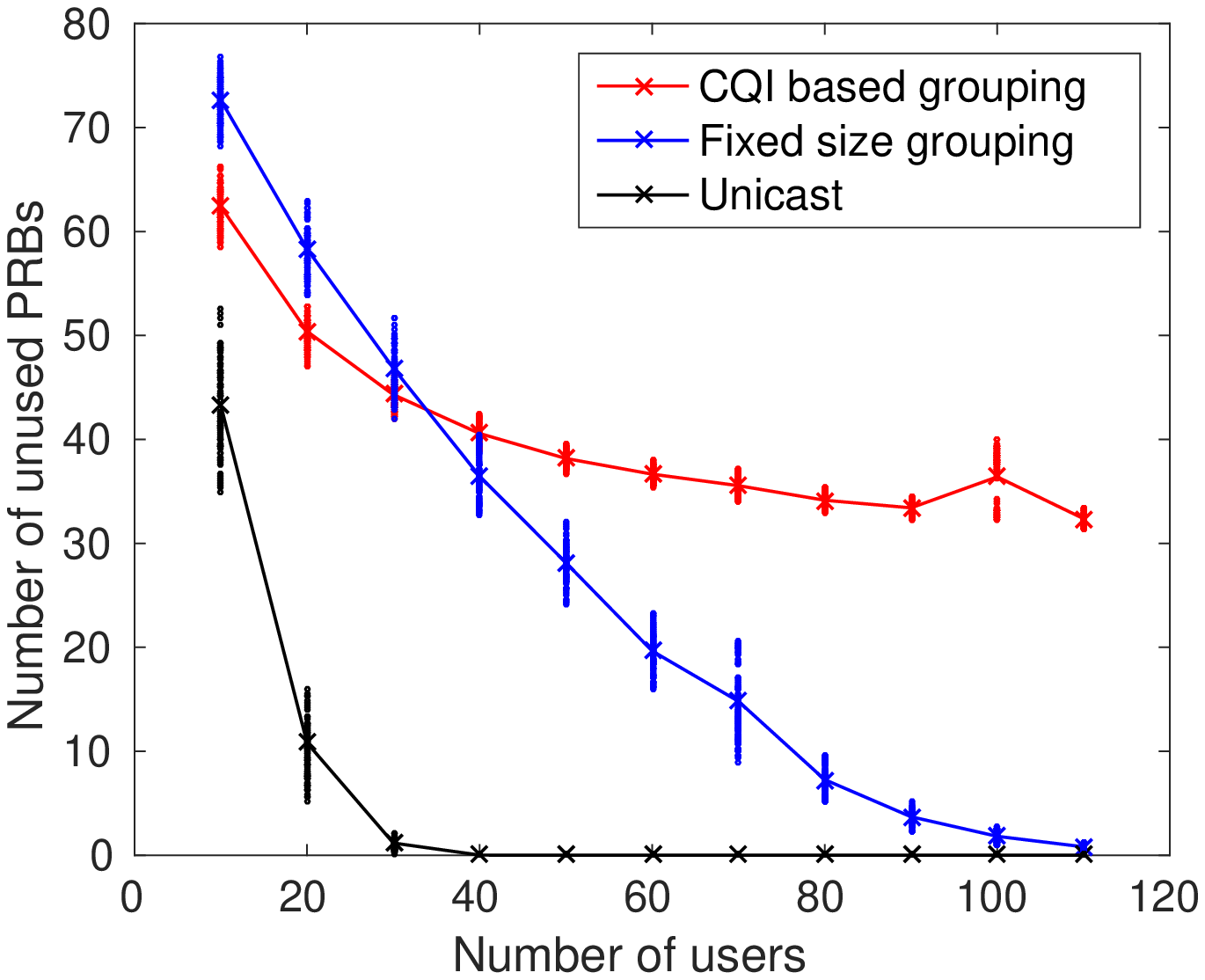}
\caption{Number of unused PRBs with LP-relaxation based scheme for different UE placements.}
\label{LP_mean}
\end{subfigure}
\caption{Performance of heuristic schemes for different UE placements.}
\end{figure*}

\begin{figure*}[htb]
    \centering
    \begin{subfigure}[t]{0.45\textwidth}
        \centering
\includegraphics[scale = 0.38]{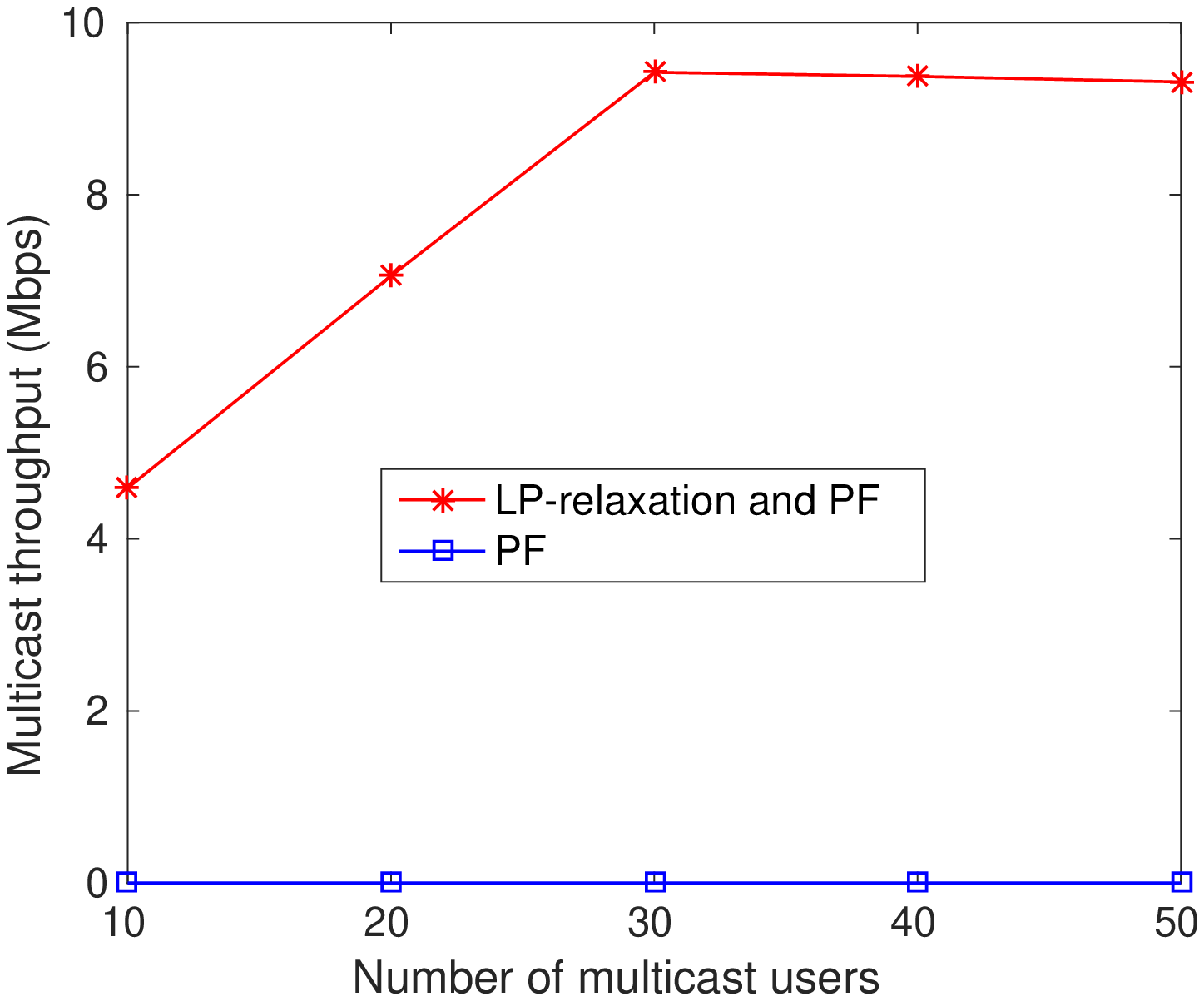} 
\caption{Average sum throughput for multicast users.}
\label{multicast_comp}
\end{subfigure}
    ~ 
    \begin{subfigure}[t]{0.45\textwidth}
        \centering
\includegraphics[scale = 0.38]{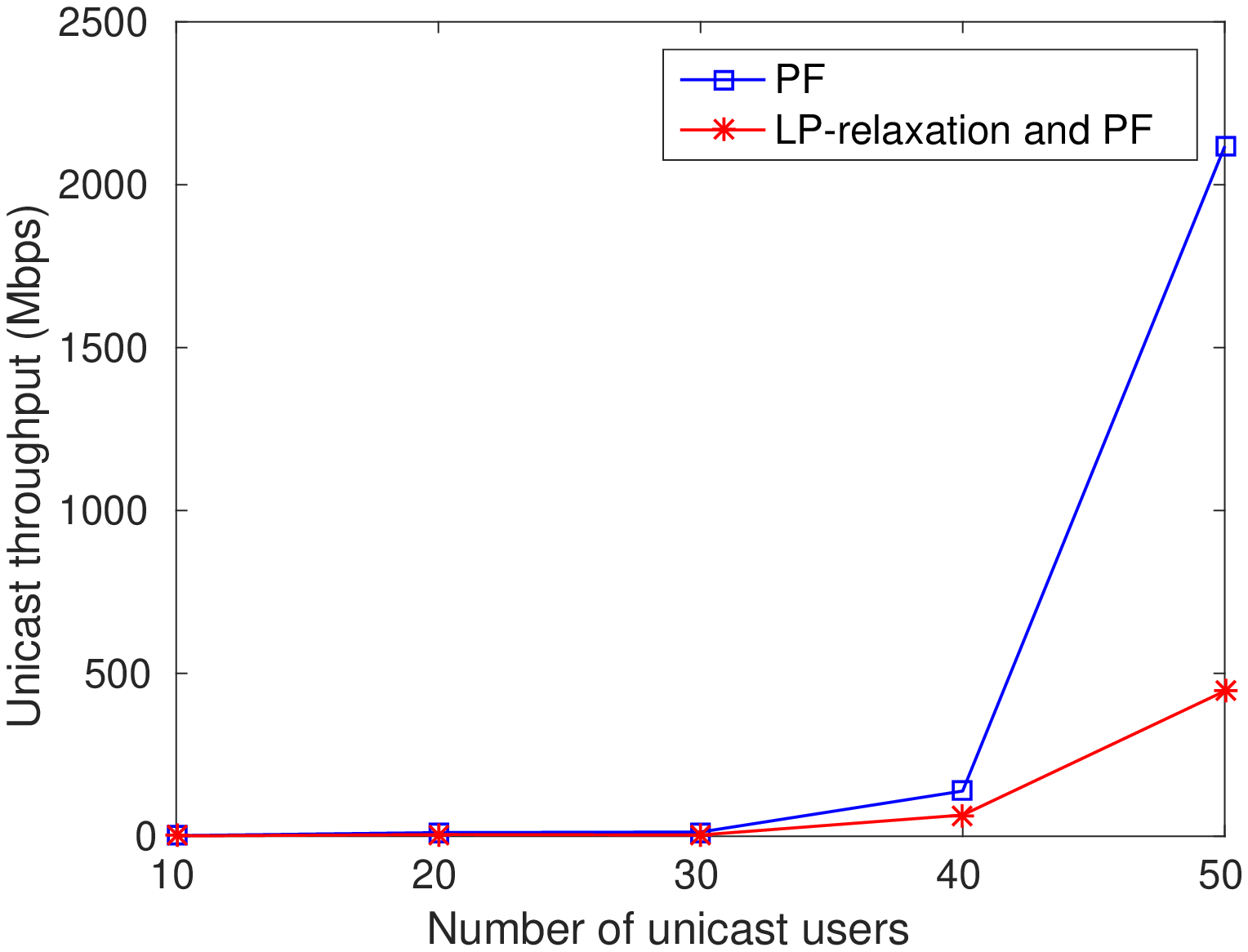}
\caption{Average sum throughput for unicast users.}
\label{unicast_comp}
\end{subfigure}
\caption{Performance comparison of LP-relaxation based scheme with PF.}
\end{figure*}

\subsection{Results}
 Fig.~\ref{varying_greedy} and Fig.~\ref{varying_LP} illustrate plots of the unutilized PRBs at the eNB against the number of UEs in the cell for greedy and LP-relaxation based
 schemes respectively. The following observations can be made from these plots:\\
\noindent 
 $\bullet$ Random grouping performs the worst. It is unable to support more than $10$ UEs.\\
 $\bullet$ The number of PRBs saved in unicast transmissions rapidly decreases to $0$ beyond $20$ UEs. Even for $20$ UEs, less than $10$ PRBs are
  saved at the eNB.\\
 $\bullet$ For fixed size grouping, greedy allocation saves greater than $10$ PRBs for up to $60$ UEs and none for
 $90$ or more UEs.
  Using LP-relaxation based allocation, more than $10$ PRBs are saved for up to $70$ UEs and none for $100$ UEs and beyond.\\
 $\bullet$ With CQI based grouping, the number of PRBs saved at the eNB is always greater than $20$ and $30$ for greedy and LP-relaxation
  based allocation respectively.
 
 Fig.~\ref{varying_infeasible_greedy} and Fig.~\ref{varying_infeasible_LP} illustrate the number of sub-frames (out of $1000$) for which the
 allocations are rendered infeasible for greedy and LP-relaxation based schemes respectively. The following can be observed from these:\\
 \noindent
 $\bullet$ Unicast transmission and random grouping quickly become completely infeasible beyond $30$ UEs for greedy allocation and beyond $40$ for LP-relaxation
  based allocation. \\ 
 $\bullet$ For the fixed size grouping and greedy allocation, up to $30$ UEs are supported with no infeasible cases. Beyond $100$ UEs,
  the allocation becomes completely infeasible. When using LP-relaxation based scheme, up to $40$ UEs are supportable without any
  infeasible cases and the allocation is never entirely infeasible for the range of the plot.\\
 $\bullet$ Number of infeasible cases for CQI based grouping is nearly zero throughout for greedy allocation and exactly zero for LP-relaxation based scheme.

 Fig.~\ref{greedy_mean} and Fig.~\ref{LP_mean} illustrate the number of PRBs saved at the eNB for different UE placements. For every $N$ number
 of UEs in the cell, $100$ different placements have been considered. Out of these, $95 \%$ closest to the mean have been plotted as
 a scatter plot. The means of the observations have also been indicated in the figures. The following conclusions can be drawn from these plots:\\
 \noindent
  $\bullet$ For varying UE placements, the number of PRBs saved at the eNB is between $\pm 5$ PRBs around the mean number of
  unused PRBs for both the grouping schemes under both resource allocation schemes.\\
  $\bullet$ Overall trend of the number of PRBs saved as the UE count increases is the same as observed in the previous plots.\par 
It can be seen from the above observations that the LP-relaxation based scheme performs better than the greedy scheme. For CQI based grouping, LP-relaxation
always saves more than $30$ PRBs on an average while for greedy, the average number of PRBs saved is less than $30$. The plots for number of infeasible cases
show that CQI based grouping always satisfies the rate requirements for all the groups in every sub-frame. Fixed size grouping can support up to $90$ UEs
while satisfying the rate requirements in at least half the sub-frames.
From these simulations, it is clear that the CQI based grouping and the LP-relaxation based scheme perform the best in grouping and resource allocation respectively. \par

In Fig.~\ref{multicast_comp} and Fig.~\ref{unicast_comp}, the average sum throughout provided by the LP-relaxation scheme is compared with that of Proportional Fair
scheme (henceforth referred to as PF) that is popularly used for resource allocation~\cite{genetic},\cite{genetic_power},\cite{mung_chiang}. For comparing the performance
of our scheme with the PF resource allocation, we make use of a scenario where both multicast and unicast UEs need to be served by the eNB. PF uses the notion of 
proportional fairness to maintain fairness among the unicast and multicast users. In order to compare our scheme with PF, we first allocate the required number of PRBs 
to the multicast groups using the proposed LP-relaxation based scheme. This is followed by a PF allocation to the unicast UEs in the cell. Fig.~\ref{multicast_comp} 
illustrates the graph of the average sum throughout achieved by the two schemes for the multicast UEs. As can be observed from the figure, even while minimizing the 
number of PRBs used, our scheme provides a significantly better throughput to the multicast UEs. The PF scheme, on the other hand, fails to fulfill the rate requirements of the multicast UEs. 
Fig.~\ref{unicast_comp} illustrates the graph of the average sum throughout achieved by the two schemes for the unicast UEs. The PF scheme provides a better throughput to 
the unicast UEs at the cost of the multicast groups. Even so, for a smaller number of UEs, our scheme performs very close to the PF scheme even after allocating sufficient 
resources to the multicast groups. \par
While PF schemes work well in a unicast only scenario, such schemes are not suitable for a rate constrained streaming scenarios that require a certain minimum rate to be 
provided to the subscribers in every sub-frame. This is clearly illustrated by the comparison results discussed in the last paragraph.

\section{Conclusions} \label{sec:conclusion}

In this paper, we have considered the problem of multicast transmission for distribution of common content to UEs in a LTE cell. For successful multicast transmission, two main challenges need to be addressed. The first challenge is that of multicast
group formation and the second challenge is that of resource allocation to the multicast groups. We have formulated the problem of resource allocation with the objective of minimizing the number of PRBs utilized while providing a certain minimum rate to all
the multicast groups. We have proved that the optimal grouping and the optimal resource allocation problems both are NP-hard and therefore, no polynomial time algorithms exist for determining their optimal solutions. NP-hardness of these problems were not
proved prior to this work. We have proposed a randomized scheme that works iteratively for estimating the optimal resource allocation. The output of the randomized scheme has been used as a benchmark for performance evaluation of the heuristic resource allocation schemes.
We have proposed two heuristic schemes for resource allocation, a greedy scheme and a LP-relaxation based scheme. We have compared the performance of the LP-relaxation based scheme to that of the randomized scheme.
The LP-relaxation scheme results in feasible resource allocation that saves nearly as many PRBs as that saved by using the randomized scheme in about one-fifth the time taken by the randomized scheme.
We have proposed two heuristic schemes for multicast group formation as well, a fixed size grouping scheme and a CQI based grouping scheme. Using extensive simulations, we have shown that using multicast transmission instead of unicast for delivering common content to a large
number of UEs simultaneously results in significant resource conservation. The proposed grouping and resource allocation schemes can act as an enhancement to MBMS in LTE. These enhancements will not only improve the performance of MBMS but will also make 
it's multicast operations more flexible and versatile.

\section{Proofs}
\subsection{Proof of Lemma~\ref{NPC1_lemma}} \label{NPC1}
The optimal resource allocation problem \Bstar was defined in Section~\ref{sec:sys_model}. Since \Bstar is an optimization problem, in order to prove that it is NP-hard, we must show the corresponding decision problem to be NP-complete.
The decision problem corresponding to \Bstar (denoted by \Cstar) is defined as follows:\par

\Cstar: Does there exist an assignment of binary variables \{$x_{ij}$\}, $i \in [L]$ and $j \in [N]$ such that (\ref{eq:rate_constr}) and (\ref{eq:matching_constr}) of \Bstar are satisfied?

\Cstar determines whether or not there exists a feasible solution of \Bstar.
In order to prove that \Bstar \ is a NP-hard problem, it is sufficient to show that \Cstar \ is NP-complete. We prove the NP-completeness of \Cstar \ by reduction from a version of the 3-partition problem (3P) defined below~\cite{partition}: \\
$\bullet$ \textbf{Input:} A set $\mathrm{Y}$, of $P=3m$ positive integers, $\{\rho_1,\rho_2, \ldots ,\rho_P\}$ such that $\frac{B}{4} < \rho_k < \frac{B}{2}$ for every $\rho_k \in \mathrm{Y}$ and $\sum_{k=1}^{P} \rho_k = mB$. \\
$\bullet$ \textbf{Problem:} Can we obtain a disjoint partition of $\mathrm{Y}$, $\{Y_1, Y_2, \ldots ,Y_m\}$ such that $\sum_{\rho_k \in Y_i}\rho_k = B$ and $|{Y_i}| = 3$ for every $Y_i, i \in \{1,2, \ldots ,m\}$ and $\bigcup_{i=1}^m Y_i = \mathrm{Y}$? \\
$\bullet$ \textbf{Output:} If the problem is feasible, the output is a suitable partition of $\mathrm{Y}$, else, the output states that the problem is infeasible. \\
The 3P problem is known to be NP-complete~\cite{partition}. Next, we show the NP-completeness of \Cstar \ by reduction from 3P.

 \begin{theorem} \label{thm:RA}
  \Cstar is a NP-complete problem.
 \end{theorem}

 \begin{proof}
 In order to prove that \Cstar is NP-complete, we first need to show that \Cstar belongs to the class NP. Given a certificate for \Cstar, we can verify in polynomial time whether or not it is a solution by checking if it
  satisfies the requirements stated in constraints (\ref{eq:rate_constr}) and (\ref{eq:matching_constr}) of \Bstar. This can be done in $O(LN)$ computations. Therefore, \Cstar $\in $ NP.\par 
  
  Having proved that \Cstar $\in$ NP, we now need to reduce 3P to an instance of \Cstar in polynomial time. The pseudo-code for the algorithm used for the said reduction is presented in Algorithm~\ref{algo:reductionRA}.
Note that, to define an instance of \Cstar, we need to state the number of groups, number of available PRBs, rate requirement of groups $(R)$ and the rates that can be achieved by the groups
in every PRB. These are defined in lines~\ref{line:groups} through~\ref{line:rate_groups} of Algorithm~\ref{algo:reductionRA} respectively. The reduction in Algorithm~\ref{algo:reductionRA} can be accomplished in $O(N)$ computations.\par
We now show that a solution for \Bstar gives us a solution for 3P as well.
 Assume that there exists a polynomial time algorithm for solving \Bstar. If we try to solve \Bstar using this algorithm, it will either give us a feasible solution or tell us that \Bstar is infeasible. We will now show how each of these outputs can be mapped
 to a corresponding solution for 3P.\par
 
  \begin{algorithm}
	\KwIn{3-partition problem with set $\mathrm{Y}$, of $P=3m$ positive integers, $\{\rho_1,\rho_2, \ldots ,\rho_P\}$ such that $\frac{B}{4} < \rho_k < \frac{B}{2} \ \forall \ \rho_k \in \mathrm{Y}$ and $\sum_{k=1}^{P} \rho_k = mB$}
	\KwOut{An instance of \Cstar with}
	
		 $L \leftarrow m$\\ \label{line:groups}
		 $N \leftarrow P$\\ \label{line:prbs}
		 $R \leftarrow B$\\ \label{line:rate}
		 $r_{ik} = r_k \leftarrow \rho_k \ \forall \ k \in \{1,2, \ldots ,P\} \ , \ i \in \{1,2, \ldots ,m\}$\\ \label{line:rate_groups}
	
	\caption{Pseudo-code for reducing 3P to \Cstar}
	\label{algo:reductionRA}
\end{algorithm} 
 
 Say that the algorithm gives us a feasible solution for \Bstar. Let the feasible solution be a matrix of binary values $[\tilde{x}_{ij}]_{i,j}$ for $i \in [L]$ and $j \in [N]$. 
 The corresponding solution for 3P can be obtained from this solution in polynomial time as follows:\par
 For every $i \in [m]$, \ $Y_i = \{\rho_j : \tilde{x}_{ij} = 1\}$. \\
 The solution thus obtained gives us a feasible solution for 3P as well. To prove this, we need to prove that:\\
 \noindent
  $\bullet$ The solution results in a disjoint partition of $\mathrm{Y}$, $\{Y_1,Y_2, \ldots ,Y_m\}$.\\
    $\bullet$ $\sum_{k \in Y_i} \rho_k = B$, for every $i$. \\
  $\bullet$ $|Y_i| = 3$ for every $i$.\\
  We shall prove these by contradiction as follows:
 \begin{enumerate}[wide = 0pt]
  \item Let's first show that the resulting solution is a disjoint partition on $\mathrm{Y}$. Suppose not. Then, one of the following two things must happen:
  \begin{enumerate}[wide = 0pt]
   \item \label{a} there exists $Y_i$ and $Y_k$ such that $Y_i \cap Y_k \neq \phi$ or,
   \item \label{b} there exists some $k$ such that $\rho_k \notin \bigcup\limits_i Y_i$.
  \end{enumerate}

  If \ref{a} is true and there exist $Y_i$ and $Y_k$ such that $Y_i \cap Y_k \neq \phi$, it means that:
  \begin{eqnarray*}
   &\exists \ j \in [P] \ \text{such that}, \tilde{x}_{ij} = 1 \ \text{and} \ \tilde{x}_{kj} = 1,\\
   &\implies \sum_l \tilde{x}_{lj} \geq 2,
  \end{eqnarray*}
which violates constraint~(\ref{eq:matching_constr}) of \Bstar. This means that $[\tilde{x}_{ij}]_{i,j}$ is not a feasible solution of \Bstar which is a contradiction.
Therefore, $Y_i \cap Y_k = \phi$ for every $i$ and $k \in [m]$. \par  
  If \ref{b} is true and there exists $k \in [P]$, such that $\rho_k \notin \bigcup_i Y_i$, it means 
  that $\tilde{x}_{ik} = 0$ for every $i$. But, we have a feasible solution of \Cstar which guarantees that the rate requirement of every group is satisfied. So,
  \begin{eqnarray*}
  &\sum_{k \in Y_i} \rho_k \geq B, \ \forall \ i \in [m], \\ 
  &\implies \sum_{j=1,j\neq k}^P \rho_j \geq mB, \implies \sum_{j = 1}^{P} \rho_j > mB,
  \end{eqnarray*}
  which is a contradiction. Hence,~\ref{b} cannot be true. Hence, the resulting solution will be a partition on $\mathrm{Y}$.
  
  \item We now show that $\sum_{k \in Y_i} \rho_k = B$, for every $i$. Suppose not. Since $[\tilde{x}_{ij}]_{i,j}$ is a feasible solution of \Bstar, we will have, $\sum_{k \in Y_i} \rho_k \geq B$, for every $i \in [m]$. Let's say that at least one of these is a strict inequality. That is,
  there exists $l \in [m]$ such that $\sum_{k \in Y_l} \rho_k > B$. This implies that $\sum_{i = 1}^{P} \rho_i > mB$,
which is a contradiction. Therefore, we will have $\sum_{k \in Y_i} \rho_k = B$, for every $i$.
  
  \item Next, we prove that $|{Y_i}| = 3$ for every $Y_i$. Let's suppose, for the sake of contradiction, that one subset, $Y_k$ has less than $3$ elements. Since the rate requirement of every group is $B$, we have, $\sum_{\rho_i \in Y_k} \rho_i \geq B$. Also, from the problem definition of 3P, we have,
  $\rho_i < \frac{B}{2}$. Since $Y_k$ can have a maximum of $2$ members, we get, $\sum_{\rho_i \in Y_k} \rho_i < B$ which is in contradiction to $\sum_{\rho_i \in Y_k} \rho_i \geq B$ above. Thus, $Y_k$ cannot have less than $3$ elements. 
  Therefore, $|{Y_i}| = 3$ for every $Y_i, \ i \in [m]$.
 \end{enumerate}
 We have now established that a feasible solution for \Bstar gives us a feasible solution for 3P as well. All that is left to complete the proof, is to show that, if \Bstar turns out to be infeasible, then, 3P has to be infeasible as well. We prove this by contradiction
 as follows: \par
 Let's assume that 3P has a feasible solution even when \Bstar is infeasible. This means that, there exists a disjoint partition of $\mathrm{Y}$, $\{Y_1, \ldots ,Y_m\}$ such that, $\sum_{\rho_k \in Y_i}\rho_k = B$ and $|{Y_i}| = 3$ for every $Y_i, \ i \in \{1,2, \ldots ,m\}$.
 This solution can be mapped to a corresponding solution for \Bstar as follows:

\begin{displaymath}
x_{ij} = \left \{ \begin{array}{ll}
1, &  \text{if} \ \rho_j \in Y_i, \\
0, & \ \text{otherwise}. \\
\end{array}
\right.
\end{displaymath}
So, for every $i$, we have:
\begin{eqnarray*}
 &\sum_{j=1}^{N} x_{ij} r_{ij} = \sum\limits_{\rho_j \in Y_i} r_j, \\
 &\implies \sum_{j=1}^{N} x_{ij} r_{ij} = \sum\limits_{\rho_j \in Y_i} \rho_j = B = R.
\end{eqnarray*}
Also, since $Y_i$'s form a disjoint partition of $\mathrm{Y}$, we will have, $\sum_{i=1}^{N} x_{ij} \leq 1$ for every $j$.
This means that $[x_{ij}]_{i,j}$ is a feasible solution for \Bstar which is a contradiction. Therefore, 3P has to be infeasible every time \Bstar is infeasible.\\ 
 Thus, a polynomial time solution for \Cstar results in a polynomial time solution for 3P as well which is not possible unless P $=$ NP. Therefore, there is no polynomial time algorithm for solving the optimal resource allocation problem $\implies$ \Cstar is 
 a NP-complete problem.
 \end{proof}

\begin{corollary}
 \Bstar is a NP-hard problem.
\end{corollary}
\begin{proof}
 The proof follows from Theorem~\ref{thm:RA}. Since the decision version of \Bstar is NP-complete, \Bstar is a NP-hard problem.
\end{proof}


\subsection{Proof of Lemma~\ref{NPC2_lemma}} \label{NPC2}
The optimal grouping problem \Dstar was defined in Section~\ref{sec:sys_model}.
Before addressing the hardness of the optimal grouping problem, we wish to point out that, given a grouping policy, $\Delta$, calculating $\overline{S}^\Delta$ in polynomial time may itself be hard.
Computing $\overline{S}^\Delta$ is non-trivial even when the channels are independent across UEs.
Hence, in this section, we shall prove the NP-hardness result assuming a Genie that can compute $\overline{S}^\Delta$ in polynomial time for any grouping policy $\Delta$. Thus, Genie is a map, $f_S : \Delta \rightarrow \mathbb{R}^+$.
However, on account of possible time and location dependent variations in channels which may induce arbitrary correlation among UE channel gains, we do not assume any specific structure on $f_S$. Thus, $f_S$ can be any arbitrary function that can be evaluated in
polynomial time.\\
We prove that \Dstar is NP-hard even when we are required to divide the UEs into just $2$ multicast groups, $G_1$ and $G_2$.
This simplified version of \Dstar that we use for proving it's NP-hardness will be referred to as \Estar and is stated below:\\
\Estar : Given a function $f_S : \Delta \rightarrow \overline{S}^{\Delta}$, determine $\{G_1^\star,G_2^\star\}$ such that $G_1^\star \cup G_2^\star = [M]$,
$G_1^\star \cap G_2^\star = \phi$ and $f_S\left( G_1^\star,G_2^\star \right) \geq f_S\left( G_1,G_2 \right)$ for every $\{G_1,G_2\}$.\par
We will prove the NP-hardness of \Estar by reduction from the Boolean Satisfiability problem (SAT). SAT was the first known NP-complete problem~\cite{clrs} and is 
defined as follows~\cite{clrs}:\\
$\bullet$ \textbf{Input:} SAT takes as input a boolean formula, $B$  with $n$ variables, $\{x_1,x_2, \ldots ,x_n\}$.\\
$\bullet$ \textbf{Problem:} Is there a consistent assignment for the variables, $\{x_1,x_2, \ldots ,x_n\}$ in $B$ such that it evaluates to TRUE? \\
$\bullet$ \textbf{Output:} If the problem is feasible, the output is a consistent assignment of binary variables $\{x_1,x_2, \ldots ,x_n\}$ that makes $B$ evaluate to TRUE.
 Otherwise, the output states that the problem is infeasible.\\
Next, we show the NP-hardness of \Estar by reduction from SAT.

\begin{theorem} \label{thm:grouping}
  \Estar \ is a NP-hard problem.
 \end{theorem}
 
 \begin{proof}
 We prove that \Estar is NP-hard by reducing SAT to an instance of \Estar.
The pseudo-code for the algorithm used for this reduction is presented in Algorithm~\ref{algo:reductionG}.
The reduction can be accomplished in $O(N^2)$ computations.
  We define the total number of multicast UEs to be $n$. The $i^\th$ UE in \Estar maps to the variable $x_i$ in SAT and $f_S$ is defined as $3 \ + \ $(the evaluation of $B$). 
  For calculating $f_S$, a TRUTH (T) evaluation of B corresponds to $1$ and a FALSE (F) evaluation equates to $0$.\par
 
 Let us now assume that there exists a polynomial time algorithm for solving \Estar. If we try to solve \Estar \ using this algorithm, it will either output a
 grouping with $f_S = 4$ or one with $f_S = 3$. Let's denote the output of the algorithm as $\{\tilde{G}_1,\tilde{G}_2\}$. We will now show how to map each of the possible outputs to a solution for SAT in polynomial time.
 \begin{enumerate}[wide = 0pt]
  \item $\bf{f_S = 4:}$ If the algorithm gives a grouping, $\{\tilde{G}_1,\tilde{G}_2\}$ with $f_S = 4$, it means that SAT is feasible. The feasible solution for SAT can be obtained as follows:
 \begin{displaymath}
x_i = \left \{ \begin{array}{ll}
\text{T}, &  \text{if} \ i^\th \text{UE} \in \tilde{G}_1 \\
\text{F}, & \text{if} \ i^\th \text{UE} \in \tilde{G}_2. \\
\end{array}
\right.
\end{displaymath}
 
To show that it's a feasible solution for SAT, we need to show the following two things: \\
$\bullet$ The assignments of $x_i$'s thus obtained are consistent and \\
$\bullet$ This assignment of $x_i$'s makes $B$ TRUE.\\
We shall prove these as follows:
\begin{enumerate}[wide = 0pt]
 \item We first prove that the output makes $B$ TRUE. Since $f_S = 4$, the evaluation of $B = 1 (\text{TRUE})$.
 \item Now we show that values assigned to $x_i$'s are consistent. Suppose not. Then, there exists some $x_k$ such that, $x_k =$ TRUE and $\bar{x}_k =$ TRUE. This means that:
 \begin{eqnarray*}
  &k^\th \ \text{UE} \ \in \tilde{G}_1 \ \text{and} \ k^\th \ \text{UE} \ \in \tilde{G}_2, \\
  &\implies \tilde{G}_1 \cap \tilde{G}_2 \neq \phi,
 \end{eqnarray*}
which is a contradiction. Therefore, the assignment of variables has to be consistent.
\end{enumerate}
Thus, a solution of \Estar with $f_S = 4$ gives a feasible solution for SAT.
\item $\bf{f_S = 3:}$ If the algorithm for \Estar \ gives a solution with $f_S = 3$, it means that SAT is infeasible, i.e. there is no consistent assignment of the variables in $B$ that can make $B$ evaluate to TRUE.
Suppose that's not true.
Say that SAT does have a feasible solution even though the solution for \Estar \ gave us $f_S = 3$. The solution of SAT can be mapped to a solution for \Estar \ as follows:
\begin{eqnarray*}
 &G_1 = \{i^\th \text{UE} : x_i = \text{T}\}, \\
 &G_2 = \{i^\th \text{UE} : x_i = \text{F}\}.
\end{eqnarray*}
Since the assignments of $x_i$'s are consistent, we have, $G_1 \cap G_2 = \phi$. Also,
\begin{eqnarray*}
 &G_1 \cup G_2 = \{i^\th \ \text{UE} \ : \{x_i = \text{T}\} \cup \{x_i = \text{F}\}\},\\
 &\implies G_1 \cup G_2 = [M].
\end{eqnarray*}
This means that the solution thus obtained is a feasible solution for \Estar \ with $f_S(G_1,G_2) = 4$ $\implies$ $f_S(G_1,G_2) > f_S(\tilde{G}_1,\tilde{G}_2)$ which is a contradiction. Therefore, when $f_S = 3$, SAT has to be infeasible.
 \end{enumerate}
  Thus, a polynomial time solution for \Estar results in a polynomial time solution for SAT as well which is not possible unless P = NP. Therefore, there is no polynomial time algorithm for solving
  \Estar i.e. \Estar is a NP-hard problem.
\end{proof}
 
   \begin{algorithm}
	\KwIn{Boolean Satisfiability problem with a boolean formula, $B$ of $n$ variables, $\{x_1,x_2, \ldots ,x_n\}$}
	\KwOut{An instance of \Estar}
	
		$M \leftarrow n$\\ \label{line:M}
		$i^\th \ \text{UE} \leftarrow x_i$\\ \label{line:users}
		$f_S \leftarrow 3 + (\text{the evaluation of} \ B)$ \label{line:f}
	
	\caption{Pseudo-code for reducing SAT to \Estar}
	\label{algo:reductionG}
\end{algorithm} 
 
 \begin{corollary}
  \Dstar is a NP-hard problem.
 \end{corollary}
 
 \begin{proof}
  The result follows from Theorem~\ref{thm:grouping}. Since \Estar, a simpler version of \Dstar is NP-hard, so is \Dstar.
 \end{proof}

 \subsection{Proof of Lemma~\ref{reward}} \label{Eproof}
 
 \begin{proof} 
 We have $s_{d^\star} \in \arg\max_{s_d} E(s_d)$ i.e. $E(s_{d^\star}) \geq E(s_d)$ for every $s_d \in \chi$. The solution for the BLP \Bstar \ corresponding to the state $s_{d^\star}$, $\{x_{ij}^\star\}_{i,j}$ is obtained as follows:
   \begin{equation*}
   x_{ij}^\star = \begin{cases}
                   {1}, & \forall \ j \in \overline{V}_{id^\star}, \\
                   {0}, & \text{otherwise}.
                  \end{cases}
  \end{equation*}
 In LTE, the rates achievable in a PRB are discrete and can take $15$ different values corresponding to the $15$ possible CQI values~\cite{seventh}. 
 The minimum rate that can be provided in a single PRB is $16$ kbps. We will denote this by $r_m$. Since the value of $E(.)$ depends on the value of $R$, two cases arise:\\
  $\bullet$ $R \leq r_m$ : In this case, we can satisfy all groups by allocating a single PRB to every group. 
  This is a trivial case and so, it is sufficient to consider the case with $R>r_m$. \\
  $\bullet$ $R > r_m$ : Before proving that $\{x_{ij}^\star\}_{i,j}$ is the optimal solution of \Bstar, we will first show that $\{x_{ij}^\star\}_{i,j}$ is a feasible solution of \Bstar. Suppose $\{x_{ij}^\star\}_{i,j}$ is not a feasible solution
  of \Bstar. This means, that there exists $i \in [L]$ such that $\sum_{j=1}^{N}x_{ij}^\star r_{ij} < R$. Then the reward of $s_{d^\star}$ will be:
  \begin{equation} \label{reward_star}
   E(s_{d^\star}) = ( N - \sum_{i \in [L]} \sum_{j\in [N]} x_{ij}^\star) - \sum_{i=1}^L \left[ R - \ell_{d^\star i} \right]^+ + q_{d^\star}.
  \end{equation}
Note that $q_{d^\star} < L$ because $\{x_{ij}^\star\}_{i,j}$ is infeasible.
Depending on the value of $\sum_{i \in [L]} \sum_{j\in [N]} x_{ij}^\star$, two cases arise: 
\begin{enumerate}[wide = 0pt]
 \item $\sum_{i \in [L]} \sum_{j\in [N]} x_{ij}^\star < N$: For this case, consider a state $s_d$ obtained from $s_{d^\star}$ by allotting one of the PRBs, $j' \in \overline{V}_{0d^\star}$
to one of the unsatisfied groups $i'$. On allocating $j'$ to $i'$, one of two things can happen: \\
$\bullet$ Rate requirement of the group $i'$ is satisfied: This means that $q_d = q_{d^\star} + 1$. The reward of the resulting $s_d$ will be:
\begin{equation*}
 E(s_d)  = E(s_{d^\star}) + (R - \ell_{d^\star i'}).
\end{equation*}
Since group $i'$ was unsatisfied in state $s_{d^\star}$, $(R - \ell_{d^\star i'}) > 0$. Therefore, $E(s_d) > E(s_{d^\star})$ which is a contradiction because $E(s_{d^\star}) \geq E(s_d)$ for every $s_d \in \chi$.\\

$\bullet$ Rate requirement of the group $i'$ is not satisfied: In this case, the reward of the state $s_d$ will be:
\begin{eqnarray*}
 E(s_d) = E(s_{d^\star}) - 1 + ( \ell_{di'} -\ell_{d^\star i'}).
\end{eqnarray*}
Here, $( \ell_{di'} -\ell_{d^\star i'})$ is the additional rate provided to group $i'$ by the PRB $j'$ which is why it can be no less than $r_m$. Since $r_m >1$,
$E(s_d) > E(s_{d^\star})$ which is a contradiction.

\item $\sum_{i \in [L]} \sum_{j\in [N]} x_{ij}^\star = N$: Here, the reward of $s_{d^\star}$ is:
\begin{equation*} 
   E(s_{d^\star}) = q_{d^\star} - \sum_{i=1}^L \left[ R - \ell_{d^\star i} \right]^+.
  \end{equation*}
  Since \Bstar is feasible, let $s_{d'}$ be a state corresponding to a feasible solution $\{x_{ij}\}_{i,j}$. The reward of $s_{d'}$ will be:
  \begin{equation*} 
   E(s_{d'}) = ( N - \sum_{i \in [L]} \sum_{j\in [N]} x_{ij}) + L > E_{s_{d^\star}},
  \end{equation*}
which is a contradiction.
 \end{enumerate}

Therefore, $\{x_{ij}^\star\}_{i,j}$ has to be a feasible solution of \Bstar. All we need to complete the proof is to show that $\{x_{ij}^\star\}_{i,j}$ is also
an optimal solution of \Bstar. We show this as follows: \par

%
%

Suppose $\{x_{ij}^\star\}_{i,j}$ is not an optimal solution of \Bstar. Let's denote the optimal solution of \Bstar by $\{\overline{x}_{ij}\}_{i,j}$ and the corresponding resource allocation state by $s_{\overline{d}}$.
Since $\{x_{ij}^\star\}_{i,j}$ is not the optimal solution, we will have, $\sum_{i\in [L]} \sum_{j \in [N]} x_{ij}^\star > \sum_{i\in [L]} \sum_{j \in [N]} \overline{x}_{ij}$
The reward of $s_{\overline{d}}$ will be:
\begin{eqnarray*}
 &E(s_{\overline{d}}) = (N -\sum_{i\in [L]} \sum_{j \in [N]} \overline{x}_{ij}) + L,\\
 &\implies E(s_{\overline{d}}) > (N -\sum_{i\in [L]} \sum_{j \in [N]} x_{ij}^\star) + L = E(s_{d^\star}),
\end{eqnarray*}
which is a contradiction. Therefore, $\{x_{ij}^\star\}_{i,j}$ is an optimal solution of the BLP \Bstar.
 \end{proof}

\ifCLASSOPTIONcompsoc

\ifCLASSOPTIONcaptionsoff
  \newpage
\fi



%

\bibliographystyle{ieeetr}
\addcontentsline{toc}{chapter}{Bibliography}
\bibliography{myrefs}


%

\begin{IEEEbiography}[{\includegraphics[width=1in,height=1.5 in,clip,keepaspectratio]{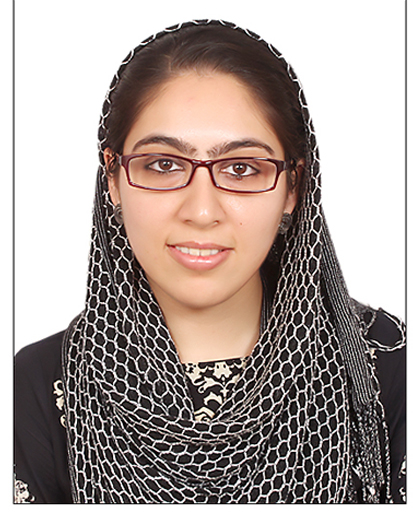}}]{Sadaf ul Zuhra} received her B.Tech degree from NIT Srinagar, India, in 2014.
 Currently, she is a research scholar in the Indian
Institute of Technology Bombay, Mumbai. Her research interests include resource allocation amd scheduling for wireless networks, video streaming and multicast communication.
\end{IEEEbiography}
\begin{IEEEbiography}[{\includegraphics[width=1in,height=1.5 in,clip,keepaspectratio]{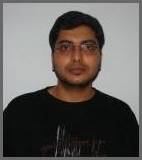}}]{Prasanna Chaporkar} received the MS degree from the Faculty of Engineering, Indian Institute of Science, Bangalore, India, in 2000,
and the PhD degree from the University of Pennsylvania, Philadelphia, Pennsylvania, in 2006. He was a
ERCIM post-doctoral fellow with ENS, Paris, France, and NTNU, Trondheim, Norway. Currently, he is an associate professor in the Indian
Institute of Technology, Mumbai. His research interests include resource allocation, stochastic
control, queueing theory, and distributed systems and algorithms.
\vspace{-17.5cm}
\end{IEEEbiography}
\begin{IEEEbiography}[{\includegraphics[width=1in,height=1.5 in,clip,keepaspectratio]{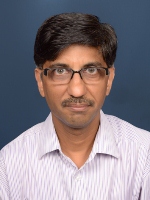}}]{Abhay Karandikar}
is currently the Director of IIT Kanpur (on leave from IIT Bombay).
In IIT Bombay, he served as Dean (Faculty Affairs) from 2017 to 2018 and the Head of the Electrical Engineering Department from 2012 to 2015.
Prof Karandikar is the founding member and chairman of Telecom Standards Development Society (TSDSI), India’s standards body for telecom setup in 2014.
His research interests include resource allocation in wireless networks, heterogeneous networks and rural broadband. Detailed biography can be found
at \url{https://www.ee.iitb.ac.in/~karandi/}.
\end{IEEEbiography}




\end{document}